\documentclass[10pt,reqno,tbtags, pdfa]{amsart}

\usepackage{caption}

\usepackage{tikz}
\usetikzlibrary{calc,matrix,patterns}
\usepackage[top=3.5cm , bottom=3cm, left=2.5cm , right=2.5cm]{geometry}
\usepackage{amsmath,amssymb,amsfonts,amsthm} 
\usepackage{mathrsfs,latexsym} 
\usepackage{mathtools}
\usepackage{thmtools}

\numberwithin{equation}{section}

\DeclareMathAlphabet{\pazocal}{OMS}{zplm}{m}{n}
\usepackage{dutchcal}

\usepackage{color}

\usepackage{epsfig}
\usepackage{graphicx}
\usepackage[all]{xy}

\usepackage{bm}
\DeclareFontFamily{OT1}{pzc}{}
\DeclareFontShape{OT1}{pzc}{m}{it}{<-> s * [1.1500] pzcmi7t}{}
\DeclareMathAlphabet{\mathpzc}{OT1}{pzc}{m}{it}

\usepackage{float}
\usepackage{wasysym}

\usepackage{todonotes}
\roman{enumi}
\newtheorem{theorem}{Theorem}[section]
\newtheorem{lemma}[theorem]{Lemma}
\newtheorem{proposition}[theorem]{Proposition}
\newtheorem{corollary}[theorem]{Corollary}

\theoremstyle{definition}
\newtheorem{definition}[theorem]{Definition}
\declaretheorem[sibling=theorem,style=definition, qed=\bell]{remark}

\numberwithin{equation}{section}

\newcommand{\abs}[1]{\left\lvert#1\right\rvert}    % absolute value
\newcommand{\norm}[1]{\left\lVert#1\right\rVert}     % norm
\newcommand{\supp}{{\mathrm{supp} }}

\newcommand{\be}{\begin{equation}}
\newcommand{\ee}{\end{equation}}
\newcommand{\ox}{\otimes}
\newcommand{\deq}{\overset{\cdot}{=}}
\newcommand{\cl}{\mathscr{C}\hspace{-1.5pt}\ell}
\newcommand{\scl}{c\hspace{-0.5pt}\ell}

\newcommand{\ccl}{\mathbb{C}\hspace{-1.5pt}\ell}
\newcommand{\diro}{\mathbf{D}}

\usepackage{enumerate}

\usepackage[backend=biber]{biblatex}
\DeclarePairedDelimiter\braket{\langle}{\rangle}
\usepackage{csquotes}
\usepackage[pdfa]{hyperref}
\usepackage{hyperxmp}
\immediate\pdfobj stream attr{/N 3}  file{eciRGB_v2.icc}
    \pdfcatalog{%
        /OutputIntents [ <<
            /Type /OutputIntent
            /S/GTS_PDFA1
            /DestOutputProfile \the\pdflastobj\space 0 R
            /OutputConditionIdentifier (eciRGB v2)
            /Info(eciRGB v2)
        >> ]
    }
\hypersetup{%
            pdftitle={Møller maps for Dirac fields in external backgrounds},
            pdfauthor={Valentino Abram and Romeo Brunetti},
             pdfsubject={Mathematical physics},
             pdfkeywords={Algebraic Quantum Field Theory, Spinor bundles, Microlocal analysis, Analysis of differential operators on vector bundles, Gauge theories, Perturbative agreement, Classical Møller map on field configurations},
             pdfcontactaddress={Via Sommarive 14},
             pdfcontactcity={\xmpquote{Povo\xmpcomma\ Trento}},
             pdfcontactpostcode={38123},
             pdfcontactcountry={Italy},
             pdfcontactemail={valentino.abram@unitn.it},
             pdflang={en},
             pdfmetalang={en},
             pdfapart=1,
             pdfcopyright={This work is licensed under a Creative Commons Attribution ShareAlike 4.0 International license},
             pdflicenseurl={https://creativecommons.org/licenses/by-sa/4.0/},,
             pdfaconformance=B,
             bookmarksopen=true,
             bookmarksopenlevel=3,
             hypertexnames=false,
                 linktocpage=true,
                 plainpages=false,
                 breaklinks
}

\renewbibmacro{in:}{}

\begin{document} 

\title[M\o ller maps for Dirac fields]{M\o ller maps for Dirac fields in external backgrounds}

\author[V. Abram]{Valentino Abram}
\address{Dipartimento di Matematica, Universit\`a di Trento, 38123 Povo (TN), Italy}
\email{valentino.abram@unitn.it}

\author[R. Brunetti]{Romeo Brunetti}
\address{Dipartimento di Matematica, Universit\`a di Trento, 38123 Povo (TN), Italy}
\email{romeo.brunetti@unitn.it}

\begin{abstract}
In this paper we  study the foundations of the algebraic treatment of classical and quantum field theories for Dirac fermions under external backgrounds following the initial contributions made by several collegues. The treatment is restricted to contractible spacetimes of globally hyperbolic nature in dimensions $d\ge 4$. In particular, we construct the classical M\o ller maps intertwining the configuration spaces of \emph{charged} and \emph{uncharged} fermions. In the last part, as a first step towards a quantization of the theory, we explore the combination of the classical M\o ller maps with Hadamard bidistributions and prove that they are involutive isomorphisms (algebraically and topologically) between suitable (formal) algebras of functionals (observables) over the configuration spaces of charged and uncharged Dirac fields.
\end{abstract}

\maketitle

\tableofcontents

\section{Introduction}
The main aim of this paper is to contribute to a by now well established framework for the algebraic treatment of quantum systems made of fermions in arbitrary backgrounds (metrics, external fields etc.) which aims at the rigorous determination of physical effects (see, \emph{e.g.}, \cite{Dimock1992, Hollands2001b, Fewster2002, Dantoni2006,  Sanders2010, dappiaggi2011, Rejzner2011, Zahn2014, Zahn15, Frob2019, Zahn2019, Murro2021, Brunetti2022}). It is a known fact that especially for strong space-time dependent external fields the mostly used theoretical frameworks suffer from several deficits. In general, the lack of space-time symmetries implies a missing privileged state (vacuum) with the related impossibility to use familiar tools as Fourier transformations and Fock spaces. All these limit terribly the ability of physicists to deduce observable effects in these quantum situations without making further drastic and simplifying assumptions. Indeed, there exists a very large literature on the subject which is full of interesting ideas, techniques and results (see, \emph{e.g.}, \cite{Eides2001,Fetodov2022}). Tipically, however, the proposals are made without referring to general and deep basic concepts and hence with a lot of \emph{ad hoc} assumptions. A possible way out of these difficulties is to refer to recent structural advances in quantum field theories using the algebraic approach. The new perspective refers to deep conceptual advancements, for instance local covariance \cite{Brunetti2003}, and technically, instead of Fourier transformations, uses its modern improvement named micro-local analysis \cite{Hormander1998}, in particular wave front sets \cite{Strohmaier2009,Brouder_2014}. There have been several recent interesting results that point towards the validity of this claim. For instance, Fr\"ob and Zahn \cite{Frob2019} have shown how to rigorously derive the trace anomaly for chiral fermions using in particular the rigorous method of Hadamard subtraction. At variance w.r.t. the literature in physics, this was done in Lorentzian spacetime, and invoking physical principles as invariance of the stress-energy tensor to show the cancellation of unwanted terms on which physicists debated for long.

Our main concern is to build up at first the classical tools that can be used later to develop the formalism towards the quantum aspects. In the present paper we concentrate basically only in the former aspects but make an initial step into the latter. We demonstrate that the classical M\o ller maps are involutive isomorphisms for the various classes of mathematical objects of pertinence for us. Indeed, at first we introduce various spinorial configurations spaces as sections of diverse bundles over semi-Riemannian spaces, and prove that the M\o ller maps are isomorphisms between the charged and uncharged spinor bundles. Thence, we extend the structure to (nonlinear) \emph{functionals} over such bundles forming (involutive) algebras. One should look at them as the (abelian) algebras of observables of the theory. 
A first step forward is done here by the construction of a \emph{Poisson} algebra. This entails at first the selection of a ``good'' subset of functionals.  The term good refers to the fact that in order to rigorously define a Poisson structure for such field theories one can make a covariant choice which is determined by the use of the Peierls' backets (see, \emph{e.g.}, \cite{Brunetti2019}). This implies the use of causal Green operators (propagators) whose kernels, seen as distributions, do not directly allow multiplications by generic functionals.  It is here that microlocal analysis appears as the right tool. The good selection is indeed made out of the desire to define the product of functionals, and their derivatives, with the propagators. A sufficient criterion for those products to exist is the H\"ormander's one based on wave front sets. Hence, the so called \emph{microcausal functionals} make their crucial appearance here. These are functionals with prescribed singularities whose wave front sets combine well with the wave front set of the propagators as to satisfy H\"ormander's criterion.

We then develop the formalism doing a first step into the quantum realm by extending the algebras to the formal algebras of deformation quantization, by changing the classical product with the use of the Hadamard prescription. 

In the course of the paper we develop geometric and analytic descriptions by adopting a practical and precise pointwise view which has the merit of being rather explicit. The resulting formalism had high advantages which we do hope compensate the heaviness of notations.

\section{Geometric and analytic preliminaries}
\subsection{Basic notions in spin geometry on pseudo-riemannian manifolds}
We begin the exposition by recalling some definitions and results in order to fix the notation used throughout the paper. We refer the reader to \cite{Tu2017}, \cite{Lawson1990} and \cite{Nicolaescu2020} for more details.

We shall work on $n$-dimensional spacetimes, that is, couples $(M,g)$ consisting of a connected, paracompact, orientable, time-orientable Hausdorff smooth manifold $M$ and a non-degenerate, pseudo-riemannian metric $g$. We shall further suppose that two additional conditions hold:
\begin{enumerate}[$(i)$]
    \item we assume $(M,g)$ is \emph{globally hyperbolic}: this entails that given a normally hyperbolic operator, this admits retarded and advanced Green operators (see, \emph{e.g.}, \cite{Bar2007,Baer2015});
    \item we also assume that $\text{dim}(M)\ge 4$: this entails that there exists a universal covering homomorphism $\xi_0\colon \mathrm{Spin}^0_{r,1} \to \text{SO}_{r,1}^0$ between the identity component of the $\mathrm{Spin}$ group $\mathrm{Spin}_{r,1}$ and the identity component of the signature $(r,1)$ of the special orthogonal group \cite[Proposition 12.1.41]{Nicolaescu2020}.
\end{enumerate}
On the manifold $M$ modelling our spacetime, we shall consider \emph{fiber bundles}, \textit{i.e.}, quadruples $(B, M, \pi, F)$ where $\pi\colon B \to M$ is a smooth, surjective map and such that there exists an open cover $\{U_\alpha\}_{\alpha\in A}$ of the base manifold $M$  and an associated collection $\{\phi_\alpha \colon \pi^{-1}(U_\alpha) \to U_\alpha \times F\}_{\alpha\in A}$ of diffeomorphisms, called \emph{trivialization of the bundle}, such that
\be
\label{Eq:1.1}
\pi_1 \circ \phi_\alpha = \pi|_{\pi^{-1}(U_\alpha)} \ \text{for every} \ \alpha\in A\ .
\ee
Notice that given a point $p\in M$, the fiber $B_p \deq \pi^{-1}(p)$ is diffeomorphic to $F$; this is the reason $F$ is called \emph{typical fiber}. 
We can consider the \emph{transition functions} of the bundle, that is, the maps
\[
\begin{split}
\phi_{\alpha\beta}\colon (U_\alpha\cap U_\beta)\times F &\to (U_\alpha\cap U_\beta)\times F\\
(p, f)\qquad &\mapsto \phi_\alpha\left(\phi_\beta^{-1}(p, f)\right)\ .
\end{split}
\]
By the condition (\ref{Eq:1.1}), we have that $\phi_{\alpha\beta}(p, f) = (p, g_{\alpha\beta}(p)(f))$ for some $g_{\alpha\beta}\colon U_\alpha\cap U_\beta \to \text{Aut}(F)$; the collection $\{g_{\alpha\beta}\}_{\alpha, \beta \in A}$ is called \emph{cocycle}.

In the context of quantum field theory, one usually deals with two kinds of fiber bundles:
\begin{enumerate}[$(i)$]
    \item \emph{vector bundles}, that is, fiber bundles whose typical fiber is a vector space and whose cocycle is such that $g_{\alpha\beta}(p)\in \text{GL}(F)$ for every $\alpha, \beta \in A$, $p\in U_\alpha\cap U_\beta$;
    \item \emph{principal $G$-bundles}, that is, fiber bundles whose typical fiber is the group $G$, whose total space $P_G$ is a right $G$-manifold with right action $r_g\colon P_G \to P_G$ and whose trivialization $\{\phi_\alpha \colon \pi^{-1}(U_\alpha) \to U_\alpha \times G\}$ consists of right $G$-equivariant maps. This entails that the cocycle $\{g_{\alpha\beta}\colon U_\alpha\cap U_\beta \to \text{Aut}(G)\}$ consists of left-translations \cite[Lemma 27.7]{Tu2017}.
\end{enumerate}
Given a principal $G$-bundle $\pi\colon P_G \to M$ and a representation $\rho\colon G \to \text{GL}(V)$ of the group on a finite-dimensional vector space $V$, we can construct the \emph{associated vector bundle} $\pi \colon P_G \times_\rho V \to M$ with typical fiber $V$. \cite[Theorem 31.9]{Tu2017} shows that there exist linear isomorphisms
\be
\label{Eq:1.2}
\begin{split}
&{\cdot}^\flat\colon \Omega_\rho^q(P_G, V)\to \Gamma\left(\wedge^q T^\ast M \otimes P_G\times_\rho V\right) \qquad {\cdot}^\sharp\colon \Gamma\left(\wedge^q T^\ast M \otimes P_G\times_\rho V\right)\to\Omega_\rho^q(P_G, V)
\end{split}
\ee
between the spaces $\Omega_\rho^q(P_G, V)$ of $V$-valued, tensorial $q$-forms of type $\rho$ on $P_G$ and the space of $q$-forms on $M$ with values in $P_G \times_\rho V$.

It is a well-known fact that gauge fields are represented by Ehresmann connections $\omega\in \Omega^1(P, \mathfrak{g})$ on a suitable principal $G$-bundle \cite{Naber2011}. These connection 1-forms can be used to induce an \emph{exterior covariant differentiation} $D\colon \Omega_\rho^q(P_G, V)\to \Omega_\rho^{q+1}(P_G ,V)$,
\be
\label{Eq:1.3}
D\varphi \deq d\varphi + \omega \cdot_\rho \varphi
\ee
where
\[
(\omega \cdot_\rho \varphi)_{p}(v_1, \dots, v_{q+1}) = \frac{1}{q!} \sum_{\sigma\in S_{q+1}}\text{sgn}(\sigma)\rho_\ast(\omega_p(v_{\sigma(1)}))\varphi_p(v_{\sigma(2)}, \dots, v_{\sigma(q+1)})\ .
\]
We can then use the exterior covariant derivative to endow the associated vector bundle $P_G \times_\rho V$ with a connection $\nabla\colon \mathfrak{X}(M)\times \Gamma(P_G\times_\rho V)\to \Gamma(P_G\times_\rho V)$, 
\[
\nabla_v s \deq (D(s^\sharp))^\flat (v)\ .
\]
This section of the associated vector bundle can be expressed locally as
\be
\label{Eq:1.4}
(\nabla_v s)(p) = \left[\sigma(p), v(p) \left( s^\sharp \circ \sigma \right)+\rho_\ast\Big(\omega_{\sigma(p)}\big({\sigma_\ast}_p(v_p)\big)\Big)s^\sharp(\sigma(p))\right]
\ee
where $\sigma\colon U \to {P_G}_U$ is a \emph{local} section of the principal $G$-bundle $P_G$, usually deemed \emph{local gauge choice}. The pullback $\sigma^\ast \omega\in \Omega^1(U, \mathfrak{g})$ is called \emph{local gauge potential}.

An important construction which can be carried out in the case of complex vector bundles is that of the \emph{conjugate bundle}. Let us recall that given a complex vector space $V$, the conjugate vector space $\overline{V}$ is a vector space which is real-isomorphic to $V$ and whose complex structure is the conjugate complex structure. The antilinear isomorphism bewteen $V$ and $\overline{V}$ shall be denoted with $C\colon \overline{V}\to V$, and it an be used to induce, starting from a linear map $L\colon V \to V$, a linear map $\overline{L}\colon \overline{V}\to \overline{V}$ by defining
\[
\overline{L}\deq C^{-1} \circ L \circ C\ .
\]
Given a basis $\{\overline{e}_i\}_{i\in I}$ of $\overline{V}$, the components $\{\overline{L}^i_j\}_{i, j\in I}$ of $\overline{L}$ are the complex conjugates of the components of $L$ in the basis $\{C \overline{e}_i\}_{i\in I}$, i.e.
\[
\overline{L}^i_j = \overline{L^i_j}
\]
If $V$ is endowed with a sesquilinear form, then one can show that $\overline{V}\simeq V^\ast$. These facts can be extended to the case of complex vector bundles in the following way: first of all, consider a vector bundle $E \overset{\pi}{\to} M$ with typical fiber $V$, whose cocycle is given by $\{g_{\alpha\beta}\}_{\alpha, \beta \in A}$; then one can construct another vector bundle with typical fiber $\overline{V}$ and whose cocycle is given by $\{\overline{g_{\alpha \beta}}\}_{\alpha, \beta \in A}$; this vector bundle is the conjugate vector bundle $\overline{E}\overset{\pi}{\to} M$. The conjugation map $C\colon \overline{V}\to V$ can be extended to a vector bundle anti-isomorphism $C\colon \overline{E}\to E$, and if $E$ is endowed with a hermitian metric $h$, then
\begin{enumerate}[(i)]
    \item $h$, which assigns to every point $p\in M$ a sesquilinear form $h_p$ on $E_p$, can be understood as a map assigning to each $p\in M$ a bilinear map $h_p\colon E_p \times \overline{E}_p\to \mathbb{C}$;
    \item we have an isomorphism $\overline{E}\simeq E^\ast$.
\end{enumerate}
Let us now consider the tangent bundle $\pi\colon TM \to M$ on our spacetime; by the assumptions on $(M,g)$ this is an oriented vector bundle endowed with a Lorentzian structure, and we can thus consider its oriented, time-oriented, pseudo-orthonormal frame bundle, that is, a principal $\text{SO}_{r,1}^0$-bundle $\pi_P \colon P_{\text{SO}_{r,1}^0}(M)\to M$ naturally associated to it. If we further assume that the second Stiefel-Whitney class $w_2(TM)$ vanishes, we can consider a \emph{spin structure} on $TM$, that is, a principal $\mathrm{Spin}_{r,1}^0$-bundle $\pi_{P'}\colon P_{\mathrm{Spin}_{r,1}^0}(M)\to M$ coupled with a 2-sheeted covering
\be
\label{Eq:1.5}
\xi\colon P_{\mathrm{Spin}^0_{r,1}}\to P_{\mathrm{SO}_{r,1}^0}
\ee
such that $\xi\left(r_g(p)\right) = r_{\xi_0(g)}\left(\xi(p)\right)$ for every $p\in P_{\mathrm{Spin}_{r,1}^0}$, $g\in \mathrm{Spin}_{r,1}^0$, with $\xi_0\colon \mathrm{Spin}_{r,1}^0\to \text{SO}_{r,1}^0$ being the universal covering map. This entails in particular that the cocycle $\{g_{\alpha\beta}\colon U_\alpha\cap U_\beta\to \text{Aut}(\text{SO}_{r,1}^0)\}$ of the principal $\text{SO}_{r,1}^0$-bundle is given by $\{\xi_0(s_{\alpha\beta})\}$, where $\{s_{\alpha\beta}\colon U_\alpha\cap U_\beta\to \text{Aut}(\mathrm{Spin}_{r,1}^0)\}$ is a cocycle for the principal $\mathrm{Spin}_{r,1}^0$-bundle.\\
By the fundamental theorem of Riemannian geometry, there exists a unique torsion-free, metric-compatible connection $\omega\in \Omega^1(P_{\text{SO}_{r,1}^0}, \mathfrak{so}_{r,1}^0)$, the Levi-Civita connection; using the covering map (\ref{Eq:1.5}), this connection can be pulled back to a connection 1-form $\omega^s\in\Omega^1(P_{\mathrm{Spin}_{r,1}^0}, \mathfrak{spin}_{r,1}^0)$, usually called \emph{spin connection}.

We can then consider the following: depending on the dimension $r+1$ of our spacetime $M$, there exists a complex representation of the Spin group $\Delta_{r,1}^{\mathbb{C}}\colon \mathrm{Spin}_{r,1}^0\to \text{GL}(V, \mathbb{C}) $ which is irreducible in the odd dimensional case and reducible in the even dimensional one \cite[\S I, Proposition 5.15]{Lawson1990}. In both cases, the representation descends from a complex representation $\rho\colon \cl_{r,1}\to \text{End}(V, \mathbb{C})$ of the Clifford algebra $\cl_{r,1}$. \cite[\S I, Theorem 5.7]{Lawson1990} shows that there exists a unique Clifford algebra representation if $r+1$ is even, while there exist two inequivalent representations if $r+1$ is odd; the induced Spin representation in the latter case doesn't depend on the chosen representation. These representations can be used to carry out the following construction:
\begin{enumerate}[$(i)$]
    \item given a complex Spin representation $\Delta_{r,1}^{\mathbb{C}}\colon \mathrm{Spin}_{r,1}^0\to \text{GL}(V, \mathbb{C})$, we can consider the \emph{complex spinor bundle}
    \be
    \label{Eq:1.6}
    S(M)\deq P_{\mathrm{Spin}_{r,1}^0}\times_{\Delta_{r,1}^{\mathbb{C}}} V
    \ee
    \item we can also consider the \emph{Clifford algebra bundle}
    \be
    \label{Eq:1.7}
    \cl(M)\deq P_{\text{SO}_{r,1}^0}\times_{\scl}\cl_{r,1}
    \ee
    where $\scl\colon \text{SO}_{r,1}^0\to \text{Aut}(\cl_{r,1})$ denotes the unique extension to the Clifford algebra $\cl_{r,1}$ of the action of $\text{SO}_{r,1}^0$ on $(\mathbb{R}^{r+1}, q_{r,1})$. It can be shown that the fiber of $\cl(M)$ at $p\in M$ is isomorphic to $\cl(T_p^\ast M, g_p)$.
\end{enumerate}
These two bundles are related by means of the \emph{Clifford module multiplication} 
\[
\mu\colon  \cl(M)\ox S(M)\to S(M)
\]
which explicitly uses the fact that $\Delta_{r,1}^{\mathbb{C}}$ descends from an algebra representation of $\cl_{r,1}$.\\
Thanks to the spin structure, this map can then be extended to yield a Clifford module multiplication
\be
\label{Eq:1.8}
\cdot_{\Gamma}\colon \mathfrak{X}( M)\ox \Gamma(S(M))\to \Gamma(S(M))
\ee
between vector fields and sections of the spinor bundle.

We shall indicate with
\be
\label{Eq:1.9}
\nabla^s\colon \mathfrak{X}(M)\ox \Gamma(S(M))\to \Gamma(S(M))
\ee
the covariant derivative induced as of (\ref{Eq:1.4}) on the spinor bundle by the spin connection $\omega^s$. This connection behaves well with respect to the Clifford multiplcation, in the sense that it is a module derivation: given $s\in \Gamma(S(M))$ and $v, u\in\mathfrak{X}(M)$ we have
\be
\label{Eq:1.10}
\nabla_u^s(v \cdot_\Gamma s) = (\nabla_u v)\cdot_\Gamma s + v \cdot_\Gamma \nabla_u^s s\ .
\ee
The Clifford multiplication can be combined with the covariant derivative $\nabla^s$ and with the musical isomorphism induced by the metric to yield a partial differential operator acting on sections of the spinor bundle, the \emph{Dirac operator}
\be
\label{Eq:1.11}
\begin{split}
\diro\colon \Gamma(S(M))&\to \Gamma(S(M))\\
s\qquad &\mapsto \quad \diro(s)\deq i\left(\cdot_\Gamma(\nabla^s(s))\right)\ .
\end{split}
\ee
Given a local orthonormal frame $\{e_i\}\subset \Gamma(TM)$ and its dual $\{e^i\}\subset \Gamma(T^\ast M)$ and considering the local section $\sigma\colon U_\alpha \to P_{\mathrm{Spin}_{r,1}^0}(M)$, (\ref{Eq:1.11}) can be written locally as
\[
\diro(s)(p) = ie^j(p) \cdot_\Gamma \left[\sigma(p),e_j(p)\left(s^\sharp\circ \sigma\right) + {\Delta_{r,1}^{\mathbb{C}}}_\ast\left(\omega_{\sigma(p)}^s(\sigma_{\ast_p}e_j(p))\right)s^\sharp(\sigma(p))\right]
\]
where we use Einstein's convention. As one can see from the local expression presented above, the Dirac operator is a first-order partial differential operator whose principal symbol $\sigma(\diro)$ is given pointwise by
\[
\sigma(\diro)\left(p, \xi_p\right)(s_p) =  \mu|_p(\xi_p, s_p)\ .
\]
Owing to the properties of Clifford multiplication and of the principal symbol, it is then evident that $\diro^2$ is a normally hyperbolic, second-order partial differential operator. Therefore, $\diro$ admits unique advanced and retarded Green operators 
\be
\label{Eq:1.12}
S_\pm \colon \Gamma_c(S(M))\to \Gamma(S(M))
\ee
i.e. $\diro\circ S_\pm = \text{id}_{\Gamma(S(M))}$, $S_\pm \circ \diro|_{\Gamma_c(S(M))} = \text{id}_{\Gamma_c(S(M))}$ and $\supp(S_\pm(s))\subseteq J^{\mp}(\supp(s))$ \cite{Baer2015}. We remind, for later purposes, that we can also define the so-called \emph{causal Green operator} (or, more simply, propagator) by posing $S\doteq S_- - S_+$.\footnote{We hope the reader shall not feel disheartened by the frequent abuse of notations, especially here for the symbols of the Green operators. The context hopefully makes clear which is which.}

By \cite[Proposition 12.1.65]{Nicolaescu2020}, we know that $S(M)$ admits a Hermitian metric such that the Clifford multiplication by $\alpha\in\Gamma(T^\ast M)$ is an Hermitian automorphism of $S(M)$. We shall denote the Hermitian metric by $(\cdot, \cdot)$, and the associated two form by $h$.
\subsection{The geometric description of charged spinors}
Following \cite{Zahn2014}, the coupling of a Dirac field with an external gauge field exploits the notion of \emph{direct product bundle}, which we recall here:
\begin{definition}
\label{Def:1.1}
Let $\pi_P\colon P \to M $ and $\pi_Q \colon Q \to M$ be two principal $G$- and $H$-bundles respectively. The product space $P\times Q$ is a principal $G\times H$-bundle over $M\times M$; by considering the diagonal $i\colon\Delta{\hookrightarrow}M\times M$ we can then consider the pullback bundle
\[
P+Q \deq i^\ast(P\times Q) = \left\{(x, p,q)\in \Delta\times(P \times Q) \ \text{s.t.} \ (x,x)=(\pi_P(p),\pi_Q(q))\right\}\ .
\]
As $\Delta$ is diffeomorphic to $M$, we have that $P+Q$ is a principal $G\times H$-bundle over $M$, called the direct product bundle.
\end{definition}
Notice that $P+Q$ can be naturally viewed as a subset of $P\times Q$; we can consider the restriction of the natural projections
\[
f_P \colon P \times Q \to P \qquad f_Q \colon P \times Q \to Q
\]
to $P+Q$. These are principal bundle homomorphisms, meaning that
\[
f_P\left(r_{(g,h)}(p,q)\right) = r_g p \qquad f_Q\left(r_{(g,h)}(p,q)\right) = r_h q \ .
\]
If the two principal bundles in Definition \ref{Def:1.1} are endowed with the connections $\omega^P \in \Omega^1(P, \mathfrak{p})$ and $\omega^Q\in\Omega^1(Q, \mathfrak{q})$ respectively, then by \cite[Proposition 6.3]{kobayashi1996foundations} we have that the 1-form
\be
\label{Eq:1.15}
\omega \deq f_P^\ast \omega^P \oplus f_Q^\ast \omega^Q
\ee
is an Ehresmann connection on $P+Q$.\\
Let us now consider two representations 
\[
\rho_1\colon G \to \text{GL}(V) \qquad \rho_2\colon H \to \text{GL}(V)
\]
such that $\rho_1(g)\rho_2(h) = \rho_2(h)\rho_1(g)$ for every $g\in G$, $h\in H$. Then one can construct the representation
\be
\label{Eq:1.16}
\begin{split}
\rho\colon G\times H &\to \text{GL}(V)\\
(g,h)&\mapsto \rho\left(g,h\right)\deq \rho_1(g)\rho_2(h)
\end{split}
\ee
whose adjoint representation is given by
\[
\begin{split}
\rho_\ast\colon \mathfrak{g}\oplus\mathfrak{h}&\to \qquad \text{End}(V)\\
(\mathbf{g}, \mathbf{h})&\mapsto {\rho_1}_\ast(\mathbf{g})+{\rho_2}_\ast(\mathbf{h})
\end{split}
\]
and consider the associated bundle
\[
(P+Q)\times_\rho V
\]
which will then be endowed with a covariant derivative.

In our case of interest we will deal with a principal $\mathrm{Spin}_{r,1}^0$-bundle $\pi_s\colon P_{\mathrm{Spin}_{r,1}^0}\to M$ coupled with a principal $G$-bundle $\pi_g\colon P_G \to M$, where $G$ is a compact Lie group which admits a representation $\rho_G\colon G\to \text{GL}(V)$ such that
\[
\Delta_{r,1}^{\mathbb{C}}(s)\rho_G(g) = \rho_G(g) \Delta_{r,1}^{\mathbb{C}}(s) \ \text{for every} \ s\in\mathrm{Spin}_{r,1}^0, \ g\in G\ .
\]
By considering the direct product bundle $P_{\mathrm{Spin}_{r,1}^0} + P_G$ endowed with the connection given by (\ref{Eq:1.15}) as well as the representation given by (\ref{Eq:1.16}) we then can construct the associated vector bundle
\[
S_G(M) \deq (P_{\mathrm{Spin}_{r,1}^0} + P_G) \times_\rho V
\]
which we shall call \emph{charged spinor bundle}. This can be endowed with the covariant derivative
\[
\nabla^{s,G}\colon \mathfrak{X}(M)\ox \Gamma(S_G(M))\to \Gamma(S_G(M)))
\]
locally given by
\[
\begin{split}
&(\nabla^{s,G}_v s)(p) = \left[\sigma(p), v(p)\left(s^{\sharp_G} \circ \sigma\right) + \rho_\ast\Big(\omega_{\sigma(p)}({\sigma_\ast}_p v(p))\Big)s^{\sharp_G}\left(\sigma(p)\right)\right] \\
& = \left[\sigma(p), v(p)\left(s^{\sharp_G} \circ \sigma\right) + \left({\Delta_{r,1}^{\mathbb{C}}}_\ast\left((f_{\mathrm{Spin}_{r,1}^0} \circ \sigma)^\ast \omega^s_p (v(p))\right) + {\rho_G}_\ast\Big((f_G \circ \sigma)^\ast \omega^G_p(v(p))\Big)\right)s^{\sharp_G}\left(\sigma(p)\right)\right]
\end{split}
\]
where $\sigma\colon U \to P_{\mathrm{Spin}_{r,1}^0}+P_G$ is a local section of the direct product bundle and
\[
\cdot^{\sharp_G}\colon \Gamma(\wedge^qT^\ast M\otimes S_G(M))\to \Omega^q_\rho(P_{\mathrm{Spin}_{r,1}^0} + P_G, V)
\]
denotes the isomorphism as of (\ref{Eq:1.2}). Notice that as $V$ is a $\cl_{r,1}$-module, even in this case we have a Clifford module multiplication 
\[
\cdot_{\Gamma, G}\colon \mathfrak{X}( M)\ox \Gamma(S_G(M)) \to \Gamma(S_G(M))
\]
which we can use to construct a Dirac operator $\diro^G$ as in (\ref{Eq:1.11}). Notice that the highest-order term is analogous to that of $\diro$, therefore the principal symbols of these two operators coincide: it follows that $\diro^G$ admits unique advanced and retarded Green operators
\be
\label{Eq:1.17}
S^G_\pm \colon \Gamma_c(S_G(M))\to \Gamma(S_G(M))\ ,
\ee
with the analogous properties of the uncharged operators as after \eqref{Eq:1.12}. Again, we can define the \emph{causal Green operator} by posing $S^G\doteq S^G_- - S^G_+$. 
%INSERT STUFF ABOUT CONJUGATION ETC
\section{The classical \texorpdfstring{M{\o}ller}{Møller} map on field configurations}
\subsection{The entwining maps between sections}
\label{Sec:2.1}
In order to compare the two theories, we need a way to relate two different spaces of functions: on one hand we have the smooth sections of the spinor bundle $\Gamma(S(M))$, while on the other we have the sections of the charged spinor bundle $\Gamma(S_G(M))$. Notice that these spaces are isomorphic, thanks to the linear isomorphisms in (\ref{Eq:1.2}), to the spaces
\[
\begin{gathered}
F\deq \left\{f\colon P_{\mathrm{Spin}_{r,1}^0} \to V, \ f \ \text{right}-\mathrm{Spin}_{r,1}^0 \ \text{equivariant}\right\} \\
F_G \deq \left\{f\colon P_{\mathrm{Spin}_{r,1}^0} + P_G \to V, \ f \ \text{right}-\mathrm{Spin}_{r,1}^0\times G \ \text{equivariant}\right\}
\end{gathered}
\]
respectively. In this first instance, we shall work with trivial principal bundles. Notice that owing to \cite[Section I]{Geroch1968} principal $\mathrm{Spin}_{r,1}^0$-bundles over non-compact spacetimes admit a global section $\sigma\colon M \to P_{\mathrm{Spin}_{r,1}^0}$, and are thus always trivial. To simplify the computations and the definition in this first case,  we make the following assumptions: 
\begin{enumerate}[(i)]
\item \emph{we assume that the principal $G$-bundle $P_G$ is trivial}; this holds, for instance, if we consider a contractible spacetime, or if we consider principal bundles $P_{U(1)}$ whose first characteristic class $c_1(P)$ vanishes.

As $P_G$ is trivial, there exists a global section $\sigma_G\colon M \to P_G$; we can use it to define the global section
\[
M\ni p\overset{\tilde{\sigma}}{\mapsto}(p, \sigma(p), \sigma_G(p))
\]
of the direct product bundle $P_{\mathrm{Spin}_{r,1}^0}+ P_G$. It follows easily that $f_{\mathrm{Spin}_{r,1}^0}\circ \tilde{\sigma} = \sigma$, $f_G\circ \tilde{\sigma} = \sigma_G$. Moreover, notice that as the principal bundles $P_{\mathrm{Spin}_{r,1}^0}$ and $P_{\mathrm{Spin}_{r,1}^0} + P_G$ are trivial, the associated fiber bundles are also trivial, with diffeomorphisms given by
\begin{alignat*}{3}
\varphi^{s}\colon (M \times \mathrm{Spin}_{r,1}^0) \times_{\Delta_{r,1}^{\mathbb{C}}} V &\to \quad M \times V \qquad \varphi^{s,G}\colon (M \times (\mathrm{Spin}_{r,1}^0 \times G)) \times_{\rho} V &&\to \quad M \times V\\
[(p,s), v]\qquad  &\mapsto (p, \Delta_{r,1}^{\mathbb{C}}(s) v) \qquad \qquad \qquad  [(p, s,g), v] &&\mapsto (p, \rho(s,g)v)\ .
\end{alignat*}
As the vector bundles are trivial, the Hermitian metrics $(\cdot, \cdot)$ on $S(M)$ and $(\cdot, \cdot)_G$ on $S_G(M)$ can be induced by a Hermitian metric $(\cdot, \cdot)_V$ on $V$; notice that thanks to \cite[Proposition 12.1.27]{Nicolaescu2020} the Hermitian metric can be chosen to be such that Clifford multiplication satisfies
\[
\rho(\phi)^\ast = -\rho(\phi^\dagger), \qquad \phi^\dagger = \tilde{\alpha}(\phi^t)
\]
where $\cdot^{t}\colon \cl_{r,1} \to \cl_{r,1}$ denotes the transpose and $\tilde{\alpha}$ is the extension of $\mathbb{R}^{r,1}\ni \mathbf{v}\to -\mathbf{v}$ to $\cl_{r,1}$. We shall denote the Hermitian matrix associated with $(\cdot, \cdot)_V$ by $(h_{ij})$.
\item As a second core assumption, \emph{we shall assume that the (global) gauge potential}
\[
\mathscr{A}(\sigma,\omega)\doteq \sigma_G^\ast \omega^G = (f_G\circ \tilde{\sigma})^\ast \omega^G \in \Omega^1(M, \mathfrak{g})
\]
\emph{is compactly supported}; its support will be denoted by $\mathrm{supp}(\mathscr{A})$.
\item \emph{We will assume that the representation $\rho_G \colon G \to \mathrm{GL}(V)$ commutes with the Clifford algebra representation inducing $\Delta_{r,1}^{\mathbb{C}}\colon \mathrm{Spin}_{r,1}^0 \to \mathrm{GL}(V)$}. Notice that this greatly reduces the freedom in the group $G$: indeed, it can be shown \cite[Theorem 4.3]{Lawson1990} that 
\[
\begin{split}
\ccl_{r+1}\simeq M(2^{\lfloor (r+1)/2 \rfloor}, \mathbb{C}) & \ \text{if} \ r+1 = 0 \mod 2 \\ 
\ccl_{r+1}\simeq M(2^{\lfloor (r+1)/2 \rfloor}, \mathbb{C})\oplus M(2^{\lfloor (r+1)/2 \rfloor}, \mathbb{C}) & \ \text{if} \ r+1 = 1 \mod 2 \qquad 
\end{split}
\]
and that an irreducible $\mathbb{C}$-module for $\cl_{r,1}$ (which descends from an irreducible module for $\ccl_{r+1}$) has complex dimension $2^{\lfloor (r+1)/2 \rfloor}$. We are therefore assuming that the image of $\rho_G$ lies in the centre of the algebra $M(2^{\lfloor (r+1)/2 \rfloor}, \mathbb{C})$, that is, $\rho_G(G)\subseteq \mathbb{C}\mathrm{id}_V$. Although it may seem rather restrictive, this case encompasses the interesting case of $G= U(1)$, that is, the case of electrically charged spinors.
\end{enumerate}
The global section $\sigma_G \colon M \to P_G$ allows us to construct the following maps:
\begin{definition}
\label{Def:2.1}
Let us define the maps $\mathcal{p}\colon F_G \to F$,
\be
\label{Eq:2.1}
F_G\ni f_g \mapsto (\mathcal{p}f_g)(p, s) \deq f_g(p, s, \sigma_G(p))
\ee
and $\mathcal{i}\colon F \to F_G$,
\be
\label{Eq:2.2}
F\ni f \mapsto(\mathcal{i}f)(p,s,g) \deq \rho_G(\tilde{g}^{-1})f(p,s)
\ee
where $\tilde{g}\in G$ is the unique group element such that $r_{\tilde{g}} \sigma_G(p) = (p, g)$, which exists as the action of $G$ on $P_G$ is free and transitive.
\end{definition}
Notice that Definition \ref{Def:2.1} is well-posed, in the sense that given $f_g\in F_G$, $\mathcal{p}f_g$ is $\text{right}-\mathrm{Spin}_{r,1}^0$ equivariant, and given $f\in F$, $\mathcal{i}f$ is $\text{right}-\mathrm{Spin}_{r,1}^0\times G$ equivariant: indeed,
\[
\begin{split}
\left(r_{\overline{s}}^\ast (\mathcal{p}f_g)\right)(p,s) &= (\mathcal{p}f_g)(p, r_{\overline{s}}s) = f_g(p, r_{\overline{s}} s, \sigma_G(p)) = (r_{(\overline{s}, \text{id}_G)}^\ast f_g)(p,s, \sigma_G(p)) \\
& = \rho(\overline{s}^{-1}, \text{id}_G)f_g(p, s, \sigma_G(p)) = \Delta_{r,1}^{\mathbb{C}}(\overline{s}^{-1})(\mathcal{p}f_g)(p,s)
\end{split}
\]
and
\[
\begin{split}
\left(r_{(\overline{s},\overline{g})}^\ast (\mathcal{i}f)\right)(p,s, g) &= (\mathcal{i}f)(p, r_{\overline{s}}s, r_{\overline{g}}g) = \rho_G(\overline{g}^{-1})\rho_G(\tilde{g}^{-1})f(p, r_{\overline{s}}s) \\
& = \rho_G(\overline{g}^{-1})\rho_G(\tilde{g}^{-1})r_{\overline{s}}^\ast f(p,s) = \rho_G(\overline{g}^{-1})\rho_G(\tilde{g}^{-1}) \Delta_{r,1}^{\mathbb{C}}(\overline{s}^{-1})f(p,s) \\
& = \Delta_{r,1}^{\mathbb{C}}(\overline{s}^{-1})\rho_G(\overline{g}^{-1})\rho_G(\tilde{g}^{-1})f(p,s) = \rho(\overline{s}^{-1}, \overline{g}^{-1})(\mathcal{i}f)(p,s,g)\ .
\end{split}
\]
We can combine the maps $\mathcal{i}$ and $\mathcal{p}$ defined in Definition \ref{Def:2.1} with the linear maps in (\ref{Eq:1.2}) to yield maps
\be
\label{Eq:2.3}
\begin{split}
\mathcal{i} \colon \Gamma(S(M))\to \Gamma(S_G(M)) &\qquad \mathcal{p}\colon \Gamma(S_G(M))\to \Gamma(S(M))\\
s\hspace{0.5cm} \mapsto \hspace{0.3cm}(\mathcal{i}s^\sharp)^{\flat_G}\hspace{0.3cm} &\qquad \hspace{1.1cm} s\hspace{0.5cm}\mapsto \hspace{0.3cm}(\mathcal{p}s^{\sharp_G})^\flat
\end{split}
\ee
which are denoted with the same symbol for the sake of notational simplicity.
Let us now endow the spaces $\Gamma(S(M))$ and $\Gamma(S_G(M))$ with the Fréchet topology induced by the families of seminorms
\be
\label{Eq:2.4}
\begin{split}
\norm{s}_{K,n}& \deq \max_{0\le i\le n} {\sup_{p\in K}\norm{\left({\nabla^s}^i s\right)(p)}_{S(M)\ox T^\ast M^{\ox i}}} \\
\norm{s}_{K,n}^G & \deq \max_{0\le i\le n} {\sup_{p\in K}\norm{\left({\nabla^{s,G}}^i s\right)(p)}_{S_G(M)\ox T^\ast M^{\ox i}}}\ .
\end{split}
\ee
We shall indicate these topological spaces as $\pazocal{E}(S(M))$ and $\pazocal{E}(S_G(M))$ respectively.
\begin{lemma}
\label{Lem:2.2}
The maps $\mathcal{i}$ and $\mathcal{p}$ defined in (\ref{Eq:2.3}) are continuous with respect to the Fréchet topologies on $\pazocal{E}(S(M))$ and $\pazocal{E}(S_G(M))$.
\end{lemma}
\begin{proof}
We need to exhibit, for each $K\subset \subset M$ and for each $n\in \mathbb{N}$, constants $k_{K,n}, k_{K,n}^G\in\mathbb{R}^+$ and families $\{(K_l, n_l)\}_{1\le l \le m}, \{(K_l^G, n_l^G)\}_{1\le l \le m^G}$ such that
\[
\norm{\mathcal{i}s}^G_{K,n}\le k_{K,n} \max_{1\le l \le m}\norm{s}_{K_l, n_l} \qquad \norm{\mathcal{p}s}_{K,n}\le k_{K,n}^G \max_{1\le l \le m^G}\norm{s}^G_{K_l^G, n_l^G}\ .
\]
The case $n=0$ is easy: indeed,
\[
\begin{split}
\norm{\mathcal{i}s}^G_{K,0} &= {\sup_{p\in K}}\norm{(\mathcal{i}s)(p)}_{S_G(M)} = \sup_{p\in K} \norm{\rho(\pi_{2}(\tilde{\sigma}(p)))(\mathcal{i}s)^{\sharp_G}(\tilde{\sigma}(p))}_{V} \\
& = \sup_{p\in K} \norm{\rho_G(\pi_2\circ \sigma_G(p))\Delta_{r,1}^{\mathbb{C}}(\pi_2 \circ \sigma(p))s^\sharp(\sigma(p))}_{V} \\ 
& \le \sup_{p\in K}\norm{\rho_G\left(\pi_2 \circ {\sigma_G(p)}\right)}_{\text{End}(V)} \sup_{p\in K}\norm{ \Delta_{r,1}^{\mathbb{C}}(\pi_2 \circ \sigma(p))s^\sharp (\sigma(p))}_V =  k_{K, 0} \norm{s}_{K, 0}
\end{split}
\]
where we use that $K\subset \subset M$ and that $\rho_G$ is a smooth representations onto a matrix algebra on $V$ finite-dimensional vector space.\\
Consider now $n=1$; in this case we compute $\nabla_{e_i}^{s,G}(\mathcal{i}s)(p)$, where $\{e_i\}_{1\le i\le \dim(M)}$ is a local pesudo-orthonormal basis for $TM$. In particular, using a section $\widehat{\sigma}\colon M \to P_{\mathrm{Spin}_{r,1}^0} + P_G$ with $\widehat{\sigma} = r_g \tilde{\sigma}$,
\[
\begin{split}
\nabla_{e_i}^{s,G}(\mathcal{i}s)(p) & = \Big[\widehat{\sigma}(p), e_i(p)\Big(\rho_G\left(\pi_G \circ g(p)\right)^{-1}s^\sharp(f_P \circ \widehat{\sigma}(p))\Big) + \Big({\Delta_{r,1}^{\mathbb{C}}}_\ast\left(\left((f_P \circ \widehat{\sigma})^\ast \omega^s\right)_p\left(e_i (p)\right)\right) \\
&\quad + {\rho_G}_\ast\left(\left((f_G \circ \widehat{\sigma}) ^\ast \omega^G\right)_p\left(e_i(p)\right)\right)\Big)(\mathcal{i}s)^{\sharp_G}(\widehat{\sigma}(p))\Big]  \\
& = \left[\widehat{\sigma}(p), \rho_G\left(\pi_G \circ g(p)\right)^{-1}\left(e_i(p)\left(s^\sharp \circ f_P \circ \widehat{\sigma}\right) + {\Delta_{r,1}^{\mathbb{C}}}_\ast\left(\left((f_P \circ \widehat{\sigma})^\ast \omega^s\right)_p(e_i(p))\right)s^\sharp(\sigma(p))\right)\right]\\
&\quad + \left[\widehat{\sigma}(p), e_i(p)\left(\rho_G\left(\pi_G \circ g(p)\right)^{-1}\right)s^{\sharp}(f_P \circ \widehat{\sigma}(p))+{\rho_G}_\ast\left(\left((f_G \circ \widehat{\sigma})^\ast \omega^G\right)_p(e_i(p))\right)(is)^{\sharp_G}(\widehat{\sigma}(p))\right]\ .
\end{split}
\]
Notice that the first term is equal to $\mathcal{i}\circ \nabla^s_{e_i}s$, while the second one can be written as
\[
\left[\widehat{\sigma}(p), {\rho_G}_\ast\left(\left(-(\pi_G \circ g)^\ast \theta_G + (f_G \circ \widehat{\sigma})^\ast \omega^G\right)_p(e_i(p))\right)(\mathcal{i}s)^{\sharp_G}(\widehat{\sigma}(p))\right]
\]
where $\theta_G$ denots the Maurer-Cartan form of $G$. Using the notation $\varepsilon_i = g(e_i, e_i)$ we have
\[
\begin{split}
&\norm{\nabla^{s,G}(\mathcal{i}s)(p)}_{S_G(M)\ox T^\ast M}^2 = \sum_{i=1}^{\dim M}\varepsilon_i\norm{\nabla^{s,G}_{e_i}(\mathcal{i}s)(p)}_{S_G(M)}^2 \\ &\le k \sum_{i=1}^{\dim M}\varepsilon_i\norm{\mathcal{i}\circ \nabla^s_{e_i} s (p)}_{S_G(M)} ^2 \\
&\qquad +k\sum_{i=1}^{\dim M}\varepsilon_i\norm{\rho\left(\pi_2 \circ \widehat{\sigma}(p)\right){\rho_G}_\ast\Big(\left(-(\pi_G \circ g)^\ast \theta_G + (f_G \circ \widehat{\sigma})^\ast \omega^G\right)_p(e_i(p))\Big) (\mathcal{i}s)^{\sharp_G}(\sigma(p))}_V^2 \ .
\end{split}
\]
Given a compact set $K \subset \subset M$, the first term can be bounded from above by
\[
k_1\sup_{p\in K}\norm{\nabla^{s}s}_{S(M)\otimes T^\ast M}^2 \le  k_1 \norm{s}_{K, 1}^2
\]
while for the second term we have the upper bound
\[
k_2\sup_{p\in K}\norm{s}_{S(M)}^2 = k_2 \norm{s}_{K,0}^2\le k_2 \norm{s}_{K,1}^2\ .
\]
Therefore we can give the estimate
\[
\norm{\mathcal{i}s}^G_{K,1} \le k_{K,1} \norm{s}_{K,1}\ .
\]
We proceed by induction: suppose we have an inequality of the form $\norm{\mathcal{i}s}^G_{K,i}\le k_{K,i}\norm{s}_{K,i}$ for $i$ up to $n-1\in \mathbb{N}$. Then by the above discussion ${\nabla^{s,G}}^n (\mathcal{i}s)$ can be written as $\mathcal{i}{\nabla^s}^n s$, which can be bounded by $k_n\norm{s}_{K,n}$, plus lower order terms in the covariant derivative of $s$ and of the map ${\rho_G}_\ast\left(\left(-(\pi_G \circ g)^\ast \theta_G + (f_G \circ \widehat{\sigma})^\ast \omega^G\right)_p(e_i(p))\right)$ which are bounded in every compact set $K\subset\subset M$ by the inductive hypothesis and the smoothness of the involved maps respectively.

An analogous proof holds for the map $\mathcal{p}\colon \Gamma(S_G(M))\to \Gamma(S(M))$: indeed,
\[
\begin{split}
\norm{\mathcal{p}s}_{K,0} &= \sup_{p\in K} \norm{(\mathcal{p}s)(p)}_{S(M)} = \sup_{p\in K}\norm{\Delta_{r,1}^{\mathbb{C}}(\pi_2 \circ \sigma (p))(\mathcal{p}s)^\sharp\left(\sigma(p)\right)}_V= \sup_{p\in K}\norm{\Delta_{r,1}^{\mathbb{C}}(\pi_2 \circ \sigma (p))s^{\sharp_G}\left(\tilde{\sigma}(p)\right)}_V  \\
& \le \sup_{p\in K} \norm{\rho_G(\pi_2 \circ \sigma_G(p))^{-1}}_{\text{End}(V)}\sup_{p\in K}\norm{\rho_G(\pi_2 \circ \sigma_G(p))\Delta_{r,1}^{\mathbb{C}}(\pi_2 \circ \sigma (p))s^{\sharp_G}\left(\tilde{\sigma}(p)\right)}_V \\
& = k^G_{K,0}\sup_{p\in K} \norm{s}_{S^G(M)}
\end{split}
\]
while for the case $n=1$ we have
\[
\begin{split}
\nabla^s_{e_i}(\mathcal{p}s)(p) & = \left[\sigma(p), e_i(p)\left(s^{\sharp_G}(\tilde{\sigma}(p))\right) + {\Delta_{r,1}^{\mathbb{C}}}_\ast\Big((\sigma^\ast \omega^s)_p (e_i(p))\Big)s^{\sharp_G}(\tilde{\sigma}(p))\right]  \\
& = \Big[\sigma(p), e_i(p)\left(s^{\sharp_G}(\tilde{\sigma}(p))\right) + {\Delta_{r,1}^{\mathbb{C}}}_\ast\Big((\sigma^\ast \omega^s)_p (e_i(p))\Big)s^{\sharp_G}(\tilde{\sigma}(p)) \\
&\quad + {\rho_G}_\ast\Big((\sigma_G^\ast \omega^G)_p(e_i(p))\Big)s^{\sharp_G}(\tilde{\sigma}(p)) - {\rho_G}_\ast\Big((\sigma_G^\ast \omega^G)_p(e_i(p))\Big)s^{\sharp_G}(\tilde{\sigma}(p))\Big]\ .
\end{split}
\]
Notice that
\[
\left[\sigma(p), e_i(p)(s^{\sharp_G} \circ \tilde{\sigma}) + \left({\Delta_{r,1}^{\mathbb{C}}}_\ast\Big((\sigma^\ast \omega^s)_p(e_i(p))\Big)+{\rho_G}_\ast\Big((\sigma_G^\ast \omega^G)_p (e_i(p))\Big)\right)s^{\sharp_G}(\tilde{\sigma}(p))\right] = \mathcal{p}\left(\nabla_{e_i}^{s,G} s\right)(p)\ .
\]
Then 
\[
\begin{split}
   & \norm{\nabla^s (\mathcal{p}s)(p)}_{S(M)\otimes T^\ast M}^2 = \sum_{i=1}^{\dim M}\varepsilon_i \norm{\nabla^s_{e_i} (\mathcal{p}s)(p)}_{S(M)}^2 \\ &\le k\sum_{i=1}^{\dim M} \varepsilon_i\norm{\mathcal{p}\left(\nabla_{e_i}^{s,G}s\right)(p)}^2_{S(M)} 
     +k\sum_{i=1}^{\dim M}\varepsilon_i\norm{\Delta_{r,1}^{\mathbb{C}}\left(\pi_2 \circ \sigma (p)\right){\rho_G}_\ast\Big((\sigma_G^\ast \omega^G)_p(e_i(p)))\Big)s^{\sharp_G}(\tilde{\sigma}(p))}_V^2\ .
\end{split}
\]
As before, given a compact $K\subset\subset M$ the first term can be bounded from above by
\[
k_1 \sup_{p\in K} \norm{\nabla^{s,G}s}_{S_G(M)\otimes T^\ast M} \le k_1 \left({\norm{s}^G_{K,1}}\right)^2
\]
and the second term can be bounded from above by
\[
k_2 \sup_{p\in K} \norm{s}_{S_G(M)}^2 \le k_2 \left({\norm{s}^G_{K,0}}\right)^2\ .
\]
By reasoning as above by induction we have the desired result.
\end{proof}
\subsection{The explicit construction of the classical \texorpdfstring{M{\o}ller}{Møller} map}
\label{Sec:2.2}
We now use these maps, in particular $\mathcal{i}\colon \pazocal{E}(S(M))\to \pazocal{E}(S_G(M))$ to give an explicit formula for the classical M{\o}ller operator, whose properties in the examined case we recall:
\begin{definition}
\label{Def:2.3}
The \emph{classical M{\o}ller map on field configurations} is a map $R_A \colon \pazocal{E}(S(M))\to \pazocal{E}(S_G(M))$ such that
\begin{enumerate}[$(i)$]
    \item $\diro^G \circ R_A = \mathcal{i} \circ \diro$ ,
    \item $R_A(s)|_{M \setminus J^+(\mathrm{supp}(\mathscr{A}))} = (\mathcal{i}s)|_{M \setminus J^+(\mathrm{supp}(\mathscr{A}))}$ .
\end{enumerate}
\end{definition}
In this case, as anticipated, the classical M{\o}ller operator has an explicit form:
\begin{theorem}
\label{Th:2.4}
The unique solution to the requirements in Definition \ref{Def:2.3} is given by
\be
R_A = \mathcal{i} - S^G_{-} \circ A \circ \mathcal{i}
\ee
where $A\colon \pazocal{E}(S_G(M))\to \pazocal{E}(S_G(M))$ acts pointwise as
\be
\label{Eq:2.6}
(As)(p) \deq ie^j(p) \cdot_{\Gamma,G} \left[\widehat{\sigma}(p), {\rho_G}_\ast\Big(\left(-(\pi_G \circ g)^\ast \theta_G + (f_G \circ \widehat{\sigma})^\ast \omega^G\right)_p(e_j(p))\Big)s^{\sharp_G}(\tilde{\sigma}(p))\right]
\ee
with $g\colon M \to \mathrm{Spin}_{r,1}^0 \times G $ such that $\widehat{\sigma} = r_g \tilde{\sigma}$.
\end{theorem}
\begin{remark}
\label{Rem:2.5}
Notice that the map $A\colon \pazocal{E}(S_G(M))\to \pazocal{E}(S_G(M)) $, as it is written in (\ref{Eq:2.6}), is gauge independent: indeed, given another section $\overline{\sigma}\colon M \to P_{\mathrm{Spin}_{r,1}^0} + P_G$ we have that $\overline{\sigma} = r_{g'} \circ \widehat{\sigma}$ where $g'\colon M \to \mathrm{Spin}_{r,1}^0 \times G$ and 
\[
f_P(r_{g'} \circ \widehat{\sigma}) = r_{\pi_{\mathrm{Spin}}(g')}f_P(\widehat{\sigma}) \qquad f_G(r_{g'} \circ \widehat{\sigma}) = r_{\pi_{G}(g')}f_G(\widehat{\sigma})\ .
\]
The following relations holds between the different pullbacks of the connection form $\omega^G$ and of the Maurer-Cartan form $\theta_G$:
\[
\begin{split}
(f_G \circ \overline{\sigma})^\ast \omega^G & = (r_{\pi_G(g')} \circ f_G \circ \widehat{\sigma})^\ast \omega^G = \text{Ad}_{\pi_G(g')^{-1}} \circ (f_G \circ \widehat{\sigma})^\ast \omega^G + (\pi_G(g'))^\ast \theta_G \ , \\
(\pi_G(g g'))^\ast \theta_G &= (r_{\pi_G(g')} \circ \pi_G(g))^\ast \theta_G = \text{Ad}_{\pi_G(g')^{-1}}\circ (\pi_G(g))^\ast \theta_G + (\pi_G(g'))^\ast \theta_G\ .
\end{split}
\]
Therefore,
\[
\begin{split}
& ie^j(p) \cdot_{\Gamma,G} \left[\overline{\sigma}(p), {\rho_G}_\ast \Big(\left(-(\pi_G(gg'))^\ast \theta_G + (f_G \circ \overline{\sigma})^\ast \omega^G\right)_p(e_j(p))\Big) s^{\sharp_G}(\overline{\sigma}(p))\right]  \\
& = ie^j(p)\cdot_{\Gamma,G} \left[\widehat{\sigma}(p),  \rho(g'(p)){\rho_G}_\ast \Big(\text{Ad}_{\pi_G(g')^{-1}} \circ \left((f_G \circ \widehat{\sigma})^\ast \omega^G - (\pi_G(g))^\ast \theta_G\right) \Big)_p(e_j(p))\rho({g'}^{-1}(p))s^{\sharp_G}(\widehat{\sigma}(p))\right]  \\
& = ie^j(p)\cdot_{\Gamma,G} \left[\widehat{\sigma}(p), {\rho_G}_\ast\Big(\left(-(\pi_G(g))^\ast \theta_G+  (f_G \circ \widehat{\sigma})^\ast \omega^G\right)_p(e_j(p)) \Big)s^{\sharp_G}(\widehat{\sigma}(p))\right]\ .
\end{split}
\]
In our case we do have a preferred gauge choice, given by the section $\tilde{\sigma}\colon M \to P_{\mathrm{Spin}_{r,1}^0}+P_G$; using that section, as $g \equiv \text{id}_G$ we then have
\[
(As)(p) = ie^j(p)\cdot_{\Gamma,G} \left[\tilde{\sigma}(p), {\rho_G}_\ast\Big(\left(\sigma_G^\ast \omega^G\right)_p(e_j(p)) \Big)s^{\sharp_G}(\tilde{\sigma}(p))\right]
\]
which, by the properties required to $\sigma_G^\ast \omega^G$, is compactly supported in $\mathrm{supp}(\mathscr{A})$. As such, in the following the gauge potential shall be the one induced by the section ${\sigma}_G$.
\end{remark}
\begin{proof}[Proof of Theorem \ref{Th:2.4}]
First of all, notice that the composition $S_-^G \circ A \circ \mathcal{i}$ is well-defined thanks to the support properties of $\sigma_G^\ast \omega^G$: indeed, we know that 
\[
S_-^G \colon \pazocal{D}(S_G(M))\to \pazocal{E}(S_G(M))
\]
and as $As$ vanishes outside of $\mathrm{supp}(\mathscr{A})$ for every $s\in \pazocal{E}({S}_G(M))$ we have that $(A \circ \mathcal{i})(s)\in \pazocal{D}(S_G(M))$ for every $s\in \pazocal{E}(S(M))$.\\
Thanks to the properties of the retarded propagator $S_-^G$, we thus have that 
\[
\supp\left((S_-^G \circ A \circ \mathcal{i})(s)\right)\subseteq J^+\left(\supp(A \circ \mathcal{i})(s)\right)\subseteq J^+\left(\mathrm{supp}(\mathscr{A})\right)
\]
and therefore the requirement (ii) in Definition \ref{Def:2.3} is fulfilled. 

To proceed further, we need to show that $\mathcal{i}\colon \pazocal{E}(S(M))\to \pazocal{E}(S_G(M))$ commutes with the Clifford multiplication by sections of the cotangent bundle. That is, we need to show that
\[
 e^i \cdot_{\Gamma, G} \mathcal{i}(u) = \mathcal{i}\left(e^i \cdot_{\Gamma} u\right)\ .
\]
We do so by investigating the result locally. Given a section $\widehat{\sigma}\colon M \to P_{\mathrm{Spin}_{r,1}^0} + P_G$ inducing a section $f_P \circ \widehat{\sigma} \colon M \to P_{\mathrm{Spin}_{r,1}^0}$, we have that locally $\mathcal{i}(u(p)) = [\widehat{\sigma}(p), \rho_G((\pi_2 \circ h)(p))^{-1}u^\sharp(f_P \circ \widehat{\sigma}(p))]$ where $h\colon M \to \mathrm{Spin}_{r,1}^0 \times G$ is such that $\widehat{\sigma} = r_h \tilde{\sigma}$. Then
\[
\begin{split}
e^i(p) \cdot_{\Gamma, G} \mathcal{i}(u)(p)& = \left[\widehat{\sigma}(p), \left(g^{ij}(p) e_j(p)\right)\cdot_V \rho_G\left((\pi_2 \circ h)(p)\right)^{-1} u^\sharp (f_P\circ\widehat{\sigma}(p))\right]  \\
& = \left[\widehat{\sigma}(p), \rho_G\left((\pi_2 \circ h)(p)\right)^{-1}\left(\left(g^{ij}(p) e_j(p)\right)\cdot_V u^\sharp (f_P\circ\widehat{\sigma}(p))\right)\right] \\
& = \mathcal{i}(e^i \cdot_\Gamma u )(p) \ .
\end{split}
\]
Thus, if we consider now $\diro^G \circ (\mathcal{i} - S_-^G \circ A \circ \mathcal{i})$ by the proof of Lemma \ref{Lem:2.2} and by the above discussion we know that
\[
\diro^G \circ \mathcal{i} = \mathcal{i} \circ \diro + A \circ \mathcal{i}
\]
and using the property $\diro^G \circ S_-^G = \text{id}_{\pazocal{D}(S_G(M))}$ of the retarded propagator associated to $\diro^G$ we have $\diro^G \circ (S_-^G \circ A \circ \mathcal{i}) = A \circ \mathcal{i}$; therefore
\[
\diro^G \circ (\mathcal{i} - S_-^G \circ A \circ \mathcal{i}) = \mathcal{i}\circ \diro\ .
\]
Therefore, also the requirement (i) in Definition \ref{Def:2.3} is fulfilled. Thus we can say that
\[
R_A = \mathcal{i} - S_-^G \circ A \circ \mathcal{i}\ .
\]
As far as the uniqueness statement is concerned, we proceed as done in \cite{Drago2017} and \cite{Ginoux2009}: given $s\in \pazocal{E}(S(M))$, let us define $R_A s \deq \psi \in \pazocal{E}(S_G(M))$; then we have that $\psi$ satisfies the following:
\[
\diro^G \psi = (\mathcal{i} \circ \diro) s \ ,\quad \psi|_{M \setminus J^+(\mathrm{supp}(\mathscr{A}))} = (\mathcal{i}s)|_{M \setminus J^+(\mathrm{supp}(\mathscr{A}))}\ .
\]
As $M$ is globally hyperbolic, it is isometric to $\mathbb{R}\times \Sigma$, with $\{a\}\times \Sigma$ Cauchy surface for every $a\in\mathbb{R}$. In particular, we can assume that $\{0\}\times \Sigma \subseteq M \setminus J^+(\mathrm{supp}(\mathscr{A}))$. Let $\{K_n\}_{n\in\mathbb{N}}$ be an invading sequence of compact sets for $\{0\}\times \Sigma$, and define
\[
\widehat{K}_n \deq D(K_n) \cap [-n, n] \times \Sigma
\]
where $D(K_n)$ denotes the Cauchy development of $K_n$. We then consider the family $\{\chi_n\}_{n\in\mathbb{N}}\subseteq \pazocal{D}(M)$, with $\chi_n \equiv 1$ on $\widehat{K}_n$. Using these family, we consider
\[
\begin{cases}
(\diro^G \circ \diro^G) \phi_n = (\mathcal{i}\circ \diro)(\chi_n s)\, & \text{on} \ M\\
\phi_n= (S_-^G \circ \mathcal{i})(\chi_n s)\, & \text{on} \ \{0\} \times \Sigma \\
\nabla^{s,G}_{\nu} \phi_n = \nabla^{s,G}_{\nu} \left((S_-^G \circ \mathcal{i})(\chi_n s)\right)\, & \text{on} \ \{0\} \times \Sigma\ .
\end{cases}
\]
Notice that $\diro^G \circ \diro^G$ is normally hyperbolic, and thus the above system admits a unique solution which depends continuously on the initial data \cite{Bar2007}. Notice in particular that on $M \setminus J^+(\mathrm{supp}(\mathscr{A}))$ the map
\[
(S_-^G \circ \mathcal{i}) (\chi_n s)
\]
is a solution of the above problem, as there the part due to the gauge potential vanishes. Therefore by the uniqueness of the solution we have
\[
\phi_n |_{M\setminus J^+(\mathrm{supp}(\mathscr{A}))} = (S_-^G \circ \mathcal{i}) (\chi_n s)|_{M \setminus J^+(\mathrm{supp}(\mathscr{A}))}\ .
\]
By the properties of $\{\widehat{K}_n\}_{n\in\mathbb{N}}$ and $\{\chi_n\}_{n\in\mathbb{N}}$ and an analogous reasoning, we have that if $m>n$ then $\phi_m = \phi_n$ on $\widehat{K}_n$; therefore by defining
\[
\phi(p)\deq \phi_n(p) \ \text{with} \ n \ \text{such that} \ p\in \widehat{K}_n 
\]
and $\psi \deq \diro^G \phi$ we have that
\[
\diro^G \psi = (\diro^G \circ \diro^G) \phi = (\mathcal{i} \circ \diro)s\ , \quad \psi|_{M \setminus J^+(\mathrm{supp}(\mathscr{A}))} = (\diro^G \circ \phi)|_{M \setminus J^+ (\mathrm{supp}(\mathscr{A}))} = (\mathcal{i}s)|_{M \setminus J^+ (\mathrm{supp}(\mathscr{A}))}\ .
\]
By the uniqueness, the $\psi$ above is unique for every $s$, and therefore $R_A$ is as well.
\end{proof}
\begin{remark}
\label{Rem:2.6}
Notice that $R_A \colon \pazocal{E}(S(M)) \to \pazocal{E}(S_G(M))$ admits both a right and left inverse. Let us consider the map
\be
\label{Eq:2.7}
\widehat{R}_A \deq \mathcal{p} + S_- \circ \mathcal{p} \circ A
\ee
and consider
\[
R_A \circ \widehat{R}_A = \mathcal{i} \circ \mathcal{p} + \mathcal{i} \circ S_- \circ \mathcal{p} \circ A -S_-^G \circ A \circ \mathcal{i} \circ \mathcal{p} - S_-^G \circ A \circ \mathcal{i} \circ S_- \circ \mathcal{p} \circ A\ .
\]
It is easy to see that $\mathcal{i}\circ\mathcal{p} = \text{id}$ and $\mathcal{p} \circ \mathcal{i} = \text{id}$ on ${\pazocal{E}(S_G(M))}$ and ${\pazocal{E}(S(M))}$ respectively, while by the previous proof we know that $A \circ \mathcal{i} = \diro^G \circ \mathcal{i} - \mathcal{i} \circ \diro$; combining these observation we conclude that
\[
R_A \circ \widehat{R}_A = \text{id} + \mathcal{i} \circ S_- \circ \mathcal{p} \circ A - S_-^G \circ A - S_-^G \circ \left(\diro^G \circ \mathcal{i} - \mathcal{i} \circ \diro\right) \circ S_- \circ \mathcal{p} \circ A\ .
\]
Due to the properties of the propagators and of the support of $A$, using the fact that $M$ is globally hyperbolic and the fact that $S_-^G \circ \diro^G = \mathrm{id}_{\pazocal{E}(S_G(M))}$ on smooth functions with past compact support we then infer that 
\[
R_A \circ \widehat{R}_A = \text{id}_{\pazocal{E}(S_G(M)} + \mathcal{i}\circ S_- \circ \mathcal{p} \circ A - S_-^G \circ A - \mathcal{i}\circ S_- \circ \mathcal{p} \circ A + S_-^G \circ \mathcal{i} \circ \mathcal{p}\circ A = \text{id}_{\pazocal{E}(S_G(M))}\ .
\]
In the same way, using the fact that
\[
(\mathcal{p} \circ A)(s) = ie^j(p)\cdot_{\Gamma,G} \left[\tilde{\sigma}(p), {\rho_G}_\ast\Big(\left(\sigma_G^\ast \omega^G\right)_p(e_j(p)) \Big)s^{\sharp_G}(\tilde{\sigma}(p))\right] = \mathcal{p}\circ \diro^G - \diro \circ \mathcal{p}
\]
one can show that
\[
\widehat{R}_A \circ R_A = \text{id}_{\pazocal{E}(S(M))}\ .
\]
We shall thus write $\widehat{R}_A$ as $R_A^{-1}$.
\end{remark}
\section{The behaviour of the classical \texorpdfstring{M{\o}ller}{Møller} map in the case of a \texorpdfstring{$U(1)$}{U(1)} gauge charge}
\subsection{The classical \texorpdfstring{M{\o}ller}{Møller} map and Green operators}
If the gauge group $G$ is $U(1)$ the entwining maps $\mathcal{i}$ and $\mathcal{p}$ defined in Section \ref{Sec:2.1} and the classical M{\o}ller map on field configurations $R_A$, whose explicit expression is given in Theorem \ref{Th:2.4}, enjoy some further properties which we now explore.

First of all, notice that if we consider the hermitian inner products on $S(M)$ and $S_G(M)$, then $\mathcal{i}$ and $\mathcal{p}$ are the adjoint of one another: indeed,
\[
\begin{split}
\left(s, (\mathcal{i}t)\right)_G (p)& = \left(\rho(\pi_2 \circ \widehat{\sigma}(p))s^{\sharp_G}(\widehat{\sigma}(p)), \rho(\pi_2 \circ \widehat{\sigma}(p))(\mathcal{i}t)^{\sharp_G}(\widehat{\sigma}(p)) \right)_V  \\
& = \left(\Delta_{r,1}^{\mathbb{C}}\left(\pi_{\mathrm{Spin}} \circ \pi_2 \circ \widehat{\sigma}(p)\right) \rho(g(p))^{-1}s^{\sharp_G}\left(\tilde{\sigma}(p)\right), \Delta_{r,1}^{\mathbb{C}}\left(\pi_{\mathrm{Spin}} \circ \pi_2 \circ \widehat{\sigma}(p) \right) \rho\left(g(p)\right)^{-1}t^{\sharp}( {\sigma}(p))\right)_V \\ 
& = \left(\Delta_{r,1}^{\mathbb{C}}\left(\pi_2 \circ \sigma (p)\right)s^{\sharp_G}\left(\tilde{\sigma}(p)\right), \Delta_{r,1}^{\mathbb{C}}\left(\pi_2 \circ \sigma(p) \right) t^\sharp({\sigma}(p))\right)_V \\
& = \left(\Delta_{r,1}^{\mathbb{C}}\left(\pi_2 \circ \sigma (p)\right)(\mathcal{p}s)^{\sharp}\left({\sigma}(p)\right), \Delta_{r,1}^{\mathbb{C}}\left(\pi_2 \circ \sigma(p) \right) t^\sharp({\sigma}(p))\right)_V  \\
&  = \left((\mathcal{p}s), t\right)(p)
\end{split}
\]
where we have used the fact that $G = U(1)$.

We now move on to the behaviour of the map $A\colon \pazocal{E}(S_G(M)) \to \pazocal{E}(S_G(M))$ defined in (\ref{Eq:2.6}). First of all, recall that $\mathfrak{u}(n)$ consists of skew-hermitian matrices, and thus $\mathfrak{u}(1)$ consists of purely imaginary complex numbers; therefore,
\[
\begin{split}
 \left(s, (At)\right)_G (p) &= \left(s(p), (At)(p)\right)_G   \\
& = \left(s^{\sharp_G}(\widehat{\sigma}(p)), i(g^{kj}(p)e_{j}(p)) \cdot_V \left({\rho_G}_\ast\left(\left(-(\pi_G \circ g)^\ast \theta_G + (f_G \circ \widehat{\sigma})^\ast \omega^G\right)_p(e_k(p))\right)t^{\sharp_G}(\widehat{\sigma}(p))\right) \right)_V  \\
& = -\left( i(g^{kj}(p)e_{j}(p)) \cdot_V s^{\sharp_G}(\widehat{\sigma}(p)), {\rho_G}_\ast\left(\left(-(\pi_G \circ g)^\ast \theta_G + (f_G \circ \widehat{\sigma})^\ast \omega^G\right)_p(e_k(p))\right)t^{\sharp_G}(\widehat{\sigma}(p)) \right)_V  \\
& = \left( i{\rho_G}_\ast\left(\left(-(\pi_G \circ g)^\ast \theta_G + (f_G \circ \widehat{\sigma})^\ast \omega^G\right)_p(e_k(p))\right)\left((g^{kj}(p)e_{j}(p)) \cdot_V s^{\sharp_G}(\widehat{\sigma}(p))\right), t^{\sharp_G}(\widehat{\sigma}(p)) \right)_V  \\
& = \left(i(g^{kj}(p)e_{j}(p)) \cdot_V \left({\rho_G}_\ast\left(\left(-(\pi_G \circ g)^\ast \theta_G + (f_G \circ \widehat{\sigma})^\ast \omega^G\right)_p(e_k(p))\right) s^{\sharp_G}(\widehat{\sigma})\right), t^{\sharp_G}(\widehat{\sigma})\right)_V  \\
& = \left((As)(p), t(p)\right)_G = \left((As), t\right)_G (p)\ .
\end{split}
\]
Let us then consider the formal adjoint ${R_A}^\ast \colon \pazocal{E}(\overline{S_G(M)})\to \pazocal{E}(\overline{S(M)})$ of the classical M{\o}ller map on field configurations, i.e.
\[
\int_M \left((R_At), s\right)_G(p) \, d\mu_g = \int_M \left(t, ({R_A}^\ast s)\right)\, d\mu_g \ \text{for every} \ t\in \pazocal{D}(S(M)), s\in \pazocal{D}(\overline{S_G(M)})\ .
\]
Notice that as $R_A\colon \pazocal{E}(S(M))\to \pazocal{E}(S_G(M))$ is linear, using the antilinear isomorphisms between a vector bundle and its conjugate bundle
\[
C\colon \overline{S(M)} \to {S(M)} \qquad C_G\colon \overline{S_G(M)}\to S_G(M)
\]
we can naturally induce a linear map ${\overline{R}_A}\colon \pazocal{E}(\overline{S(M)})\to \pazocal{E}(\overline{S_G(M)})$, $\overline{R}_A \deq  C_G^{-1}\circ R_A \circ C$, whose formal adjoint is given by ${\overline{R}_A}^\ast\colon \pazocal{E}(S_G(M))\to\pazocal{E}(S(M))$. In particular, it holds that ${\overline{R}_A}^\ast = \overline{R_A^\ast}$.
In light of the previous equalities, we can write
\[
{\overline{R}_A}^\ast = \mathcal{p} - \mathcal{p}\circ A \circ S_+^G \qquad R_A^\ast = C^{-1} \circ {\overline{R}_A}^\ast \circ C_G\ .
\]
Notice that the antilinear isomorphisms can be used to define
\[
\overline{\mathcal{i}}\colon \pazocal{E}(\overline{S(M)})\to \pazocal{E}(\overline{S_G(M)}) \qquad 
\overline{\mathcal{p}}\colon \pazocal{E}(\overline{S_G(M)})\to \pazocal{E}(\overline{S(M)}) \ .
\]
These maps, as well as the map $\overline{R}_A\colon \pazocal{E}(\overline{S(M)})\to\pazocal{E}(\overline{S_G(M)})$ satisfy the same results as the ones previously proven.
\begin{proposition}
\label{Prop:3.1}
The advanced and retarded propagators of $\diro^G$ and $\diro$ are related by
\be
\label{Eq:3.1}
S_-^G = R_A \circ S_- \circ \mathcal{p} \qquad S_+^G = \mathcal{i} \circ S_+ \circ {\overline{R}_A}^\ast\ .
\ee
where the composition of maps appearing on the right-hand sides are restricted to $\pazocal{D}(S_G(M))$.
\end{proposition}
\begin{proof}
Let us denote with $\widehat{S}_\pm^G$ the operators on the right-hand sides of the equalities in (\ref{Eq:3.1}). It is easy to see that on $\pazocal{D}(S_G(M))$ it holds that
\[
\diro^G \circ \widehat{S}^G_- = \diro^G \circ R_A \circ S_- \circ \mathcal{p}\overset{\text{Def.} \ \ref{Def:2.3}}{=} \mathcal{i} \circ \diro \circ S_- \circ \mathcal{p}= \mathcal{i} \circ \mathcal{p}  = \text{id}
\]
and
\[
\begin{split}
\widehat{S}_-^G \circ \diro^G & = R_A \circ S_- \circ \mathcal{p} \circ \diro^G  = R_A \circ S_- \circ \diro \circ R_A^{-1} = R_A \circ R_A^{-1} \\
& = \text{id}\ .
\end{split}
\]
The same relations for $\widehat{S}_+^G$ can be obtained in the following way: we know that ${\diro^G}^\ast = C_G^{-1}\circ \diro^G \circ C_G$ and that ${\widehat{S}_-}^{G^\ast} = C_G^{-1} \circ \widehat{S}_+^G \circ C_G$; therefore given any $u\in \pazocal{D}(S_G(M))$, $v\in\pazocal{D}(\overline{S_G(M)})$
\[
\begin{split}
 \int_M\left(u, v\right)_G\, d\mu_g &=\int_M \left(\diro^G \widehat{S}^G_- u, v\right)_G\, d\mu_g = \int_M \left(\widehat{S}_-^G u, C_G^{-1} \circ \diro^G \circ C_G v\right)_G\, d\mu_g  \\
& = \int_M \left(u, C_G^{-1} \circ \widehat{S}_+^G \circ \diro^G \circ C_G v\right)\, d\mu_g\ .
\end{split}
\]
As the above equality holds for any functions $u\in\pazocal{D}(S_G(M))$, $v\in\pazocal{D}(\overline{S_G(M)})$, we infer that $\widehat{S}^G_+ \circ \diro^G = \text{id}_{\pazocal{D}(S_G(M))}$. The other equality can be obtained in a similar fashion.

To conclude the proof, it suffices to show that
\[
\supp(\widehat{S}_\pm ^G (u)) \subseteq J^{\mp}(\supp(u)) \ \text{for every} \ u\in \pazocal{D}(S_G(M))\ .
\]
We proceed for the retarded propagator $\widehat{S}_-^G$, as the proof for the advanced one is entirely analogous. In particular, we shall prove that
\[
M\setminus J^+(\supp(u)) \subseteq M \setminus \supp(\widehat{S}_-^G(u))\ .
\]
To this end, it suffices to show that $M \setminus \supp(S_-^G (u)) \subseteq M \setminus \supp(\widehat{S}_-^G(u))$; therefore, let $p\notin \supp(S_-^G (u))$, and let us consider two Cauchy surfaces $\Sigma_1, \Sigma_2 \subseteq M$ such that
\begin{enumerate}[$(i)$]
    \item $\Sigma_2\subseteq J^+(\Sigma_1)$;
    \item $\Sigma_1 \cap \Sigma_2 = \varnothing$;
    \item $\left(\mathrm{supp}(\mathscr{A}) \cup \{p\} \cup \supp(u)\right)\cap J^+(\Sigma_1) = \varnothing$.
\end{enumerate}
We then consider a smooth function $\varphi\in\pazocal{E}(M)$ such that
\[
\varphi \equiv 1 \ \text{on} \ J^-(\Sigma_1), \qquad \varphi \equiv 0 \ \text{on} \ J^+(\Sigma_2)
\]
which we use to define the maps
\begin{alignat*}{3}
\varphi\colon \pazocal{E}(S(M))& \to \pazocal{E}(S(M)) \qquad \qquad \qquad \qquad  \varphi_G\colon \pazocal{E}(S_G(M))&& \to\pazocal{E}(S_G(M))\\
    s\quad &\mapsto (\varphi s)(p) \deq \varphi(p)s(p)\qquad \qquad \qquad \qquad t&&\mapsto (\varphi_G t)(p) \deq \varphi(p) t(p)
\end{alignat*}
as well as the analogous maps $1-\varphi$ and $1-\varphi_G$. Notice that $\mathcal{i}\circ \varphi = \varphi_G \circ \mathcal{i}$ and $\mathcal{p} \circ \varphi_G = \varphi \circ \mathcal{p}$ and that $\supp(\varphi \circ S_- \circ \mathcal{p} (u))$ is compact. Then 
\[
\begin{split}
    (\widehat{S}_-^G u) & =\left( R_A \circ S_- \circ \mathcal{p}\right)(u)\\ &= \left(R_A \circ \left(1-\varphi + \varphi\right) \circ S_- \circ \mathcal{p}\right)(u) \\ & = \left(R_A \circ \varphi \circ S_- \circ \mathcal{p}\right)(u) + \left(R_A \circ (1-\varphi) \circ S_- \circ \mathcal{p}\right)(u)\\ 
    &= (\mathcal{i} - S_-^G \circ A \circ \mathcal{i}) \circ (\varphi \circ S_- \circ \mathcal{p})(u) +\left( R_A \circ (1-\varphi) \circ S_- \circ \mathcal{p} \right)(u)  \\
    & = \left(S_-^G \circ \diro^G \circ \mathcal{i} \circ \varphi \circ S_- \circ \mathcal{p} \right) (u) - \left(S_-^G \circ A \circ \mathcal{i} \circ \varphi \circ S_- \circ \mathcal{p}\right)(u) +  \left(R_A \circ (1-\varphi) \circ S_- \circ \mathcal{p}\right)(u)  \\
    & = \left(S_-^G \circ \mathcal{i} \circ \diro \circ \varphi \circ S_- \circ \mathcal{p}\right)(u) + \left(R_A \circ (1-\varphi) \circ S_- \circ \mathcal{p}\right)(u)  \\
    & = \left(S_-^G \circ \mathcal{i}\circ \diro(\varphi) \circ S_- \circ \mathcal{p}\right)(u) + \left(S_-^G \circ \mathcal{i} \circ \varphi \circ \diro \circ S_- \circ \mathcal{p}\right)(u) + \left(R_A \circ (1-\varphi) \circ S_- \circ \mathcal{p}\right)(u)  \\
    & = \left(S_-^G \circ \mathcal{i} \circ \diro(\varphi) \circ S_- \circ \mathcal{p}\right)(u) + S_-^G (u) + \left(R_A \circ (1-\varphi) \circ S_- \circ \mathcal{p}\right)(u)\ .
\end{split}
\]
Now, notice that when evaluated at $p$, the above expression is zero: indeed, we know that $p\notin \supp(S_-^G (u))$, and due to the properties of $\varphi\in\pazocal{E}(M)$ we also know that $\diro(\varphi) \equiv 0 $ on $J^-(\Sigma_1)$ and $J^+(\Sigma_2)$, and therefore $(\mathcal{i} \circ \diro(\varphi) \circ S_- \circ \mathcal{p})(u)$ is supported in $J^+(\Sigma_1)\cap J^-(\Sigma_2)$. But then $S_-^G((\mathcal{i} \circ \diro(\varphi) \circ S_- \circ \mathcal{p})u)$ is supported in the causal future of that set; as $p\in J^-(\Sigma_1)$ we then have that the first term vanishes at $p$. Let us then consider the last term: we have
\[
\left(\mathcal{i} \circ (1-\varphi) \circ S_- \circ \mathcal{p}\right)(u) = \left((1-\varphi_G) \circ \mathcal{i} \circ S_- \circ \mathcal{p}\right)(u)
\]
which is equal to zero when evaluated at $p$, and 
\[
-\left(S_-^G \circ A \circ \mathcal{i} \circ (1-\varphi) \circ S_- \circ \mathcal{p}\right)(u) = -\left(S_-^G \circ A \circ (1-\varphi_G) \circ \mathcal{i} \circ S_- \circ \mathcal{p}\right)(u)
\]
which vanishes as $\mathrm{supp}(\mathscr{A}) \cap \supp(1-\varphi) = \varnothing$. Thus $p\notin \supp(\widehat{S}_-^G (u))$, and we conclude.
\end{proof}
\begin{corollary}
\label{Cor:3.2}
The causal propagators of $\diro^G$ and $\diro$ are related by
\be
\label{Eq:3.2}
S^G = R_A \circ S \circ {\overline{R}_A}^\ast\Big|_{\pazocal{D}(S_G(M))}
\ee
\end{corollary}
\begin{proof}
Let us define $o\deq -S_- \circ \mathcal{p} \circ A$; then
\[
\begin{split}
R_A - \mathcal{i} &= R_A - R_A \circ R_A^{-1} \circ \mathcal{i}  = R_A \circ \left(\text{id}-R_A^{-1}\circ \mathcal{i}\right) = R_A \circ\left(\mathcal{p}\circ \mathcal{i} - \left(\mathcal{p} + S_- \circ \mathcal{p} \circ A\right) \circ \mathcal{i} \right)  \\
& = R_A \circ (-S_- \circ \mathcal{p}\circ A \circ \mathcal{i}) = R_A \circ o \circ \mathcal{i}\ .
\end{split}
\]
Analogously, $R_A - \mathcal{i}= \mathcal{i} \circ o \circ R_A$; moreover, we also have the same relations concerning the formal adjoints:
\[
{\overline{R}_A}^\ast -\mathcal{p} = {\overline{R}_A}^\ast \circ o^\dagger \circ \mathcal{p} \qquad {\overline{R}_A}^\ast -\mathcal{p} = \mathcal{p} \circ o^\dagger \circ {\overline{R}_A}^\ast
\]
where $o^\dagger \deq -A \circ \mathcal{i} \circ S_+$. Now, using these equalities we have on $\pazocal{D}(S_G(M))$
\[
R_A \circ (S_- - S_+) \circ {\overline{R}_A}^\ast = \left(R_A \circ S_- \circ (\mathcal{p} + \mathcal{p} \circ o^\dagger \circ\overline{ R_A}^\ast) - (\mathcal{i} + R_A \circ o \circ \mathcal{i}) \circ S_+ \circ {\overline{R}_A}^\ast\right)\ .
\]
Proposition \ref{Prop:3.1} then entails that on $\pazocal{D}(S_G(M))$ we have that 
\[
\begin{split}
  R_A \circ S \circ {\overline{R}_A}^\ast & = S_-^G+ R_A \circ S_- \circ \mathcal{p} \circ (-A \circ \mathcal{i} \circ S_+)\circ {\overline{R}_A}^\ast - S_+^G - R_A \circ o \circ \mathcal{i} \circ S_+ \circ {\overline{R}_A}^\ast  \\
  & = S^G + R_A \circ o \circ i \circ S_+ \circ {\overline{R}_A}^\ast - R_A \circ o \circ i \circ S_+ \circ {\overline{R}_A}^\ast  = S^G\ .
\end{split}
\]
\end{proof}
\subsection{The classical \texorpdfstring{M{\o}ller}{Møller} map and Hadamard bidistributions}
\label{Sec:4.2}
Having assessed the properties of the classical M{\o}ller map on field configurations when coupled with the Green operators associated to the free and uncharged Dirac operators $\diro$ and $\diro^G$, we now examine the behaviour of Hadamard bidistributions when coupled with the M\o ller maps. 

First of all, let us recall that the algebras associated to the Dirac field need to account for both the spinor and cospinor field, that is, sections of both the vector bundle $S(M)$ and its conjugate bundle $\overline{S(M)}$. In order to do so, one considers the Whitney sum of the two vector bundles $S(M)\oplus \overline{S(M)}$ and $S_G(M)\oplus \overline{S_G(M)}$, which we shall denote with $S^{\oplus}(M)$ and $S^{\oplus}_G(M)$ respectively. A section $u$ of $S^{\oplus}(M)$ can be then understood as a couple $(u_1,u_2)$ with $u_1\in\pazocal{E}(S(M))$ and $u_2\in\pazocal{E}(\overline{S(M)})$; the same holds for sections of $S^{\oplus}_G(M)$.

The hermitian metric on $S(M)$, which can be understood as a bilinear map
\[
\pazocal{E}(S(M)) \times \pazocal{E}(\overline{S(M)})\ni u,v\mapsto (u,v)\in \pazocal{E}(M)
\]
can be used to induce a symmetric and bilinear map (denoted with the same symbol)
\[
\pazocal{E}(S^\oplus(M))\times \pazocal{E}(S^\oplus(M)) \ni u,v \mapsto (u,v) \deq (v_1, u_2)+(u_1, v_2) \in \pazocal{E}(M)\ .
\]
One can also define an involution on the space of sections $\pazocal{E}(S^{\oplus}(M))$ by using the conjugation maps $C\colon \overline{S(M)}\to S(M)$ and $C^{-1}\colon S(M)\to \overline{S(M)}$:
\[
\pazocal{E}(S^{\oplus}(M))\ni u=(u_1, u_2) \mapsto u^\ast \deq (C u_2, C^{-1} u_1 )\in\pazocal{E}(S^{\oplus}(M))\ .
\]
Using the Dirac operator $\diro\colon \pazocal{E}(S(M))\to \pazocal{E}(S(M))$ and its adjoint $\diro^\ast = C^{-1} \circ \diro \circ C$ as well as the causal propagators $S\colon \pazocal{D}(S(M))\to\pazocal{E}(S(M))$ and $S^\ast$ we construct the operators
\[
\diro^\oplus \deq \diro \oplus -\diro^\ast \qquad S^\oplus \deq S \oplus -S^\ast
\]
Notice that $S^\oplus$, which is the causal propagator for $\diro^\oplus$, is formally self-adjoint: indeed, given $u,v\in\pazocal{D}(S^\oplus (M))$ we have
\[
\int_M (S^\oplus u, v)\, d\mu_g = \int_M (S u_1, v_2) - (v_1, S^\ast u_2)\, d\mu_g = \int_M -(u_1, S^\ast v_2) + (S v_1, u_2)\, d\mu_g = \int_M (u, S^\oplus v)\, d\mu_g\ .
\]
Therefore, the distribution $S^\oplus \in \pazocal{D}'(S^\oplus(M) \boxtimes S^\oplus(M))$ given by the Schwartz kernel theorem,
\[
S^\oplus (u, v)\deq \int_M (S^\oplus u,v)\, d\mu_g
\]
is symmetric. Analogous extensions can be made in the case of the charged spinor bundle $S_G(M)$.

The M{\o}ller map $R_A\colon \pazocal{E}(S(M))\to\pazocal{E}(S(M))$ as well as its conjugate $\overline{R}_A\colon \pazocal{E}(\overline{S(M)})\to \pazocal{E}(\overline{S(M)})$, which satisfy Definition \ref{Def:2.3} and Theorem \ref{Th:2.4} (with suitable modifications), can be combined into one M{\o}ller map $\mathcal{R}_A\colon \pazocal{E}(S^\oplus(M))\to\pazocal{E}(S_G^\oplus(M))$,
\[
\mathcal{R}_A \deq R_A \oplus \overline{R}_A = i^\oplus - {S_-^G}^\oplus \circ A^\oplus \circ i^\oplus\ . 
\]
The same can be done with the formal adjoints of the M{\o}ller maps, yielding
\[
{\mathcal{R}_A}^\ast = \overline{R}_A^\ast \oplus  R_A^\ast \ .
\]
We now recall the definition of Hadamard bidistribution for spinor fields.
\begin{definition}
\label{Def:3.3}
A distribution $\omega\in\pazocal{D}'(S^\oplus(M) \boxtimes {S^\oplus(M)})$ satisfies the Hadamard two-point condition if given any two sections $u,v\in\pazocal{D}(S^\oplus(M))$ we have:
\begin{enumerate}[$(i)$]
    \item $\omega((\diro^\oplus u), {v}) = 0$;
    \item $\omega(u, v) + \omega(v,u) = i S^\oplus (u, {v})$;
    \item we require that
    \[
    \text{WF}(\omega) = \left\{(x, y,  \xi_x, -\xi_y)\in T^\ast M^2 \setminus z(M^2) \ | \ (x,\xi_x) \sim (y, \xi_y) \ \mathrm{or} \ x=y, \xi_x=\xi_y, \ \xi_x \triangleright 0\right\}
    \]
    where $\xi_x \triangleright 0$ means that $\xi_x$ is future-directed and lightlike, and $(x,\xi_x) \sim (y, \xi_y)$ means that $x$ can be connected to $y$ by means of a future-directed lightlike geodesic $\gamma$ such that $\xi_x$ is the cotangent vector of $\gamma$ at $x$ and $\xi_y$ is the cotangent vector of $\gamma$ at $y$.
\end{enumerate}
\end{definition}
Notice that such distributions do exist; see for instance \cite{Zahn2014} and \cite{Murro2021}.
\begin{proposition}
\label{Prop:3.4}
If $\omega\in \pazocal{D}'(S^\oplus (M)\boxtimes {S^\oplus(M)})$ is a distribution satisfying the Hadamard two-point condition, then
\be
\label{Eq:3.3}
\omega_G(\cdot, \cdot)\deq  \omega({\mathcal{R}_A}^\ast \cdot, {\mathcal{R}_A}^\ast\cdot)
\ee
is a distribution in $\pazocal{D}'(S_G^\oplus(M)\boxtimes {S_G^\oplus(M)})$ satisfying the Hadamard two-point condition.
\end{proposition}
\begin{proof}
First of all, notice that $\omega_G$ defined as in (\ref{Eq:3.3}) is well-defined: indeed, $\overline{R}_A^\ast$ and ${R_A}^\ast$ are continuous with respect to the inductive limit topology of $\pazocal{D}(S_G(M))$ and $\pazocal{D}(\overline{S_G(M)})$, being linear maps which are sequentially continuous. Moreover, they evidently map compactly supported smooth functions to compactly supported smooth functions. These result directly translate to the map ${\mathcal{R}_A}^\ast$. Let us now prove that the three requirements are satisfied:
\begin{enumerate}[$(i)$]
    \item it is easy to see that $\omega_G({\diro^G}^\oplus u, {v})=0$; indeed, by Definition \ref{Def:2.3} and Theorem \ref{Th:2.4} we have
    \[
    R_A^\ast \circ {\diro^{G}}^\ast=(\diro^G \circ R_A)^\ast = (\mathcal{i} \circ \diro)^\ast = \diro^\ast \circ \mathcal{i}^\ast
    \]
    and therefore ${\mathcal{R}_A}^\ast \circ {\diro^G}^\oplus = (\diro \circ \mathcal{p})\oplus (-\diro^\ast \circ \mathcal{i}^\ast)$; this entails that
    \[
    \omega_G({\diro^G}^\oplus  u, {v}) = \omega\left({{\mathcal{R}_A}}^\ast( {\diro ^G}^\oplus  u), {\mathcal{R}_A}^\ast {v}\right) = \omega(\diro^\oplus ((\mathcal{p} \oplus \mathcal{i}^\ast) u), {\mathcal{R}_A}^\ast { v}) = 0\ .
    \]
    \item We know that $\omega_G(u,{v}) = \omega({\mathcal{R}_A}^\ast u, {\mathcal{R}_A}^\ast {v})$, and by using Corollary \ref{Cor:3.2} we have
    \[
    \begin{split}
    \omega_G(u,{v}) + \omega_G(v,{u}) & = \omega({\mathcal{R}_A}^\ast u, {\mathcal{R}_A}^\ast {v}) +\omega({\mathcal{R}_A}^\ast v, {\mathcal{R}_A}^\ast {u})\\ &= iS^\oplus({\mathcal{R}_A}^\ast u, {\mathcal{R}_A}^\ast {v}) \\ &= i \int_M \left(S^\oplus({\mathcal{R}_A}^\ast u), {\mathcal{R}_A}^\ast {v}\right)\, d\mu_g  \\
    & = i\int_M \left(S \circ \overline{R}_A^\ast u_1, R_A^\ast v_2\right)-\left(\overline{R}_A^\ast v_1, S^\ast \circ {R}_A^\ast u_2 \right)\, d\mu_g \\
    & = i\int_M \left((R_A \circ S \circ \overline{R}_A^\ast) u_1, v_2\right) -\left( v_1, (\overline{R}_A \circ S^\ast \circ {R}_A^\ast) u_2 \right)\, d\mu_g \\
    & = i \int_M \left({S^G}^\oplus u,{v}\right)_G\, d\mu_g = i {S^G}^\oplus(u,{v})\ .
    \end{split}
    \]
    \item We know that $\omega_G$ is a bisolution of ${\diro^G}^\oplus$, and that the principal symbol of ${\diro^G}^\oplus$ coincides with that of $\diro^\oplus$; therefore by \cite[Theorem 6.1.1]{Duistermat1972} we know that the wavefront set of $\omega_G$ coincides with that of $\omega$, being determined by the Hamiltonian flow associated to the principal symbol of ${\diro^G}^\oplus$.
\end{enumerate}
\end{proof}
\section{The classical \texorpdfstring{M{\o}ller}{Møller} map on the observable algebras}
\subsection{The Poisson \texorpdfstring{$\ast$}{*}-algebras of observables}
Our goal is to pass the classical M{\o}ller map on field configurations $R_A$, presented in Definition \ref{Def:2.3} and Theorem \ref{Th:2.4}, to the algebras of observables of the charged and uncharged Dirac field. We briefly recall the construction of said topological algebras, as presented in \cite{Rejzner2011}, \cite{Zahn2014} and in Section 2 of \cite{Brunetti2022}.\\
Let $E\overset{\pi}{\to} M$ be a vector bundle with typical fiber $V$, be it either $S(M)\oplus \overline{S(M)}$ or $S_G(M)\oplus \overline{S_G(M)}$, endowed with a symmetric bilinear metric $h$, and let us consider the exterior algebra of $\pazocal{E}(E)$, that is, the graded algebra
\[
\wedge^\bullet \pazocal{E}(E) = \bigoplus_{p\in\mathbb{N}} \wedge^p \pazocal{E}(E)\ .
\]
We can consider the spaces of homogeneous elements $\wedge^p \pazocal{E}(E)$ as embedded into $\Gamma(M^p, E^{\boxtimes p})\simeq \overline{\Gamma(M, E)^{{\otimes}p}}$; then using the usual Fréchet topology (uniform convergence of all derivatives on compact sets) on $\Gamma(M^p, E^{\boxtimes p})$ we define the spaces of $p$-antisymmetric sections 
\[
\pazocal{E}^a(M^p, E^{\boxtimes p}) \deq \overline{\wedge^p \pazocal{E}(E)}
\]
as well as the configuration space
\[
\mathcal{C}(E) \deq{ \widehat{\bigoplus}_{p\in \mathbb{N}} \pazocal{E}^a(M^p, E^{\boxtimes p})}
\]
where $\widehat{\bigoplus}$ denotes the algebraic direct sum. Notice that the involution defined in Section \ref{Sec:4.2} can be extended to an involution ${\cdot}^\ast \colon\mathcal{C}(E)\to\mathcal{C}(E)$ by requiring the behaviour
\[
(u_1 \wedge \cdots \wedge u_p)^\ast \deq u_p^\ast \wedge \cdots \wedge u_1^\ast
\]
on homogeneous elements in $\wedge^p \pazocal{E}(E)$, by extending the above by continuity to $\pazocal{E}^a(M^p, E^{\boxtimes p})$ and by linearity to $\mathcal{C}(E)$.

We are interested in \emph{antisymmetric functionals} on the space of sections $\pazocal{E}(E)$; these can be interpreted as a sequence $\{F_p\}_{p\in\mathbb{N}}$ of linear and continuous functionals on $\{\wedge^p \pazocal{E}(E)\}_{p\in\mathbb{N}}$, that is, a sequence of elements such that
\[
F_p \in \mathcal{F}^p (E) \deq {\pazocal{E}^a}^\prime(M^p, E^{\boxtimes p}) \ \ \text{for all} \  p\in\mathbb{N}\ ,
\]
where the $(\cdot)^\prime$ means the strong topological dual.

Thus we define the space of \emph{fermionic functionals} as
\[
\mathcal{F}(E) \deq \prod_{p\in\mathbb{N}} \mathcal{F}^p(E)\ .
\]
There exists a duality pairing between $\mathcal{F}(E)$ and $\mathcal{C}(E)$ given by
\[
\braket{F, u} \deq \sum_{p\in\mathbb{N}} \braket{F_p, u_p} \ \text{for all} \ F \in \mathcal{F}(E), \ u\in\mathcal{C}(E)\ .
\]
Notice that the sum is finite and thus always well-defined, $\mathcal{C}(E)$ being an algebraic direct sum. We can endow $\mathcal{F}(E)$ with the weak topology $\tau_\sigma$, that is, the topology given by the family of seminorms $\{p_u\}_{u\in\mathcal{C}(E)}, \ p_u(F) = \abs{F(u)}$, thus making it a locally convex topological vector space which happens to be nuclear and sequentially complete.\\
$\mathcal{F}(E)$ can be endowed with an antisymmetric, pointwise product initially defined on homogeneous elements in $\wedge^p \pazocal{E}(E)$ by
\[
(F \wedge G)_p(u_1 \wedge \cdots \wedge u_p) \deq \sum_{\sigma \in S_p} \mathrm{sgn}(\sigma)\sum_{k=0}^p \frac{1}{k! (p-k)!} F_k(u_{\sigma(1)}\wedge \cdots \wedge u_{\sigma(k)})G_{p-k}(u_{\sigma(k+1)}\wedge \cdots \wedge u_{\sigma(p)})\ .
\]
The object above is then extended by linearity and continuity to elements in $\pazocal{E}^a(M^p, E^{\boxtimes p})$, thus yielding a well-defined object in $\mathcal{F}(E)$. Moreover, said product is continuous with respect to the topology on $\mathcal{F}(E)$. $\mathcal{F}(E)$ is also naturally endowed with an involution ${\cdot}^{\ast}\colon \mathcal{F}(E)\to \mathcal{F}(E)$,
\be
\label{Eq:4.8}
\{F_p\}_{p\in\mathbb{N}}\mapsto \left\{F_p^\ast\right\}_{p\in\mathbb{N}}, \qquad (F_p^\ast)(u_p) \deq \overline{F_p(u_p^\ast)}\ .
\ee

One can then consider derivatives of fermionic functionals in the following way: given $F_p\in\mathcal{F}^p(E)$, $p\ge 1$, the left derivative of $F_p$ in the direction $h\in \pazocal{E}(E)$ is defined on $\wedge^{p-1}\pazocal{E}(E)$ as 
\[
d_h F_p(u) \deq F_p(h\wedge u)
\]
and is then extended by continuity to $\pazocal{E}^a(M^{p-1}, E^{\boxtimes p-1})$, thus yielding a linear and continuous map $d_h F_p \colon \pazocal{E}^a(M^{p-1}, E^{\boxtimes p-1})\to \mathbb{C}$, i.e. $d_h F_p\in \mathcal{F}^{p-1}(E)$. One can then extend the map $d_h \colon \mathcal{F}^{p}(E)\to\mathcal{F}^{p-1}(E)$ to the whole algebra $\mathcal{F}(E)$ by considering
\[
d_h \colon F = \{F_p\}_{p\in\mathbb{N}} \mapsto d_h F \deq \{d_h F_p\}_{p\in\mathbb{N}}\ .
\]
It is easy to see that for any $h\in\pazocal{E}(E)$, $d_h$ is a graded derivation. One can also consider higher order derivatives by iterating the left derivative: given $F_p\in\mathcal{F}^p(E)$, $k\le p$ and $h_1, \dots, h_k \in\pazocal{E}(E)$, we define for $u\in \wedge^{p-k} \pazocal{E}(E)$
\[
d^{k}_{h_1, \dots, h_k} F_p(u) = F_p(h_k \wedge \cdots \wedge h_1 \wedge u)
\]
and then proceed by continuity as before, obtaining for each $F\in \mathcal{F}(E)$ a jointly continuous map
\[
d^k F \colon \pazocal{E}(E)^k \times \mathcal{C}(E) \to \mathbb{C} 
\]
which is easily seen to be multilinear and alternating in the first $k$ entries, that is, equivalently,  a continuous map
\[
F^{(k)}\colon \pazocal{E}^a(M^k, E^{\boxtimes k}) \times \mathcal{C}(E) \to \mathbb{C}
\]
which is linear in the first entry. In particular, notice that these can be considered as an $\mathcal{F}(E)$-valued (compactly supported) distributional section of ${E^\ast}^{\boxtimes k}\to M^k$, that is, an object of $\pazocal{D}'(M^k, E^{\boxtimes k}) \widehat{\otimes}_\pi \mathcal{F}(E)$, where $\widehat{\otimes}_\pi$ denotes the completion of the tensor product in the projective topology\footnote{As both spaces are nuclear, the completion is independent on the chosen topology; we choose the projective topology just to fix one.}.

In order to proceed with the quantization, one needs to restrict the $\ast$-algebra of fermionic functionals in order to endow it with a suitable $\star$-product. To do so, one needs to be able to control the wavefront set of derivatives of the relevant functionals; in particular, we consider the set of \emph{microcausal} fermionic functionals $\mathcal{A}(E)\subseteq \mathcal{F}(E)$ consisting of those functionals $F\in\mathcal{F}(E)$ such that
\[
\text{WF}(F^{(n)}_{u})\subseteq \Xi_n \deq T^\ast \overset{\cdot}{M^n} \setminus \bigcup_{(p_1, \dots, p_n)\in M^n} \left(\overline{V_{p_1}^+} \times \cdots \times \overline{V_{p_n}^+}\right) \cup  \left(\overline{V_{p_1}^-} \times \cdots \times \overline{V_{p_n}^-}\right) \ \mathrm{for \ every} \ u\in\mathcal{C}(E)\ .
\]
We endow this space with the initial locally convex topology induced by the family of linear maps $\{\ell_{k, u}\}$,
\[
F \overset{\ell_{k, u}}{\longmapsto} \begin{dcases} \braket{F, u}\in\mathbb{C}\, & k=0\\
F^{(k)}_u\in{\pazocal{E}_{\Xi_k}^a}^\prime(M^k, E^{\boxtimes k})\, & k\ge 1
\end{dcases}
\]
indexed by an integer $k\in\mathbb{N}$ and a function in the configuration space $u\in\mathcal{C}(E)$, and where ${\pazocal{E}_{\Xi_k}^a}^\prime(M^k, E^{\boxtimes k})$ denotes the inductive limit
\[
\varinjlim{\pazocal{E}_{\Gamma_{k,n}}^a}^\prime(M^k, E^{\boxtimes k})
\]
with $\{\Gamma_{k,n}\}_{n\in\mathbb{N}}$ a sequence of closed cones in $T^\ast M^k$ such that $\Gamma_{k,n} \subset \overset{\circ}{\Gamma}_{k,n+1}$ and $\cup_n \Gamma_{k,n} = \Xi_k$, and  where ${\pazocal{E}_{\Gamma_{k,n}}^a}^\prime(M^k, E^{\boxtimes k})$ is endowed with the usual H\"ormander topology. Using the causal propagator $S\in\pazocal{D}'(E^{\boxtimes 2}, M^2)$ one is then able to endow $\mathcal{A}(E)$ with a Peierls' bracket. First of all, given any homogeneous functional $F\in\mathcal{A}^p(E)$, we define the object $(S \ast F^{(1)})$,
\be
\label{Eq:4.4}
(S \ast F^{(1)})(u) \deq \int_{M} ({S(x,y), F^{(1)}_u(y)})_h\, d\mu_g(y)
\ee
which is well-defined thanks to the wavefront set properties of both the causal propagator $S$ and of $F^{(1)}_u$, as
\[
\mathrm{WF}(S) = \left\{(x,y,\xi_x, -\xi_y)\in T^\ast M^2 \setminus z(M^2) \ | \ (x, \xi_x) \sim (y, \xi_y) \right\}\ .
\]
Moreover, notice that this object is actually a smooth function: indeed, using \cite[Theorem 8.2.13]{Hormander1998} one obtains that
\[
\mathrm{WF}((S \ast F^{(1)}_u))\subseteq \mathrm{WF}_M(S) \cup \mathrm{WF}'(S)\circ \mathrm{WF}(F^{(1)}_u) = \varnothing \ \mathrm{for \ all} \ u\in \pazocal{E}^a(M^p, E^{\boxtimes p})\ .
\]
Therefore, we can consider it as an object in $\pazocal{E}(E^\ast)\,{\simeq\ } \pazocal{E}(E)$, where the isomorphism is due to the existence of the symmetric bilinear metric $h$. We can then compute, for any other homogeneous functional $G\in\mathcal{A}^q(E)$, the object
\[
G^{(1)} \wedge (S \ast F^{(1)})\in \mathcal{A}^{p+q-2}(E)
\]
which is defined on an element $u_{1}\wedge \cdots \wedge u_{p+q-2}$
\[
\begin{split}
G^{(1)}&\wedge (S \ast F^{(1)})(u_1 \wedge \cdots \wedge u_{p+q-2})\\  &\deq  (-1)^{q+1}\sum_{\sigma\in S_{p+q-2}}\mathrm{sgn}(\sigma) G\left((S \ast F^{(1)}_{u_{\sigma(q)}\wedge \cdots \wedge u_{\sigma(q+p-2)}}) \wedge u_{\sigma(1)} \wedge \cdots \wedge u_{\sigma(q-1)}\right)
\end{split}
\]
and then extended to the whole of $\pazocal{E}^a(M^{p+q-2}, E^{\boxtimes p+q-2})$ by continuity. This procedure is then extended to non-homogeneous functionals by considering $G^{(1)} \wedge (S \ast F^{(1)}) = \left\{\left(G^{(1)} \wedge (S \ast F^{(1)})\right)_p\right\}_{p\in\mathbb{N}}$,
\be
\label{Eq:4.5}
\left(G^{(1)} \wedge (S \ast F^{(1)})\right)_p(u) \deq \sum_{k=0}^{p} \frac{1}{k!(p-k)!} \left((G^{(1)})_k \wedge \left(S\ast (F^{(1)})_{p-k}\right)\right)(u), \qquad u\in \pazocal{E}^a(M^p, E^{\boxtimes p}) \ .
\ee
Then the Peierls' bracket is given by $\{F, G\}_S \deq G^{(1)} \wedge (S \ast F^{(1)})$; notice that on homogeneous elements $F\in\mathcal{A}^p(E)$, $G\in\mathcal{A}^q(E)$ we have the desired graded anticommutativity
\[
\left\{G,F\right\}_S = -(-1)^{qp}\left\{F,G\right\}_S
\]
as well as the graded Jacobi identity.
\subsection{Deformation quantization and the classical M\o ller maps}
Starting from $\mathcal{A}(E)$, the quantization proceeds in the following way: we consider the $\ast$-algebra $\mathcal{A}(E)[\![\hbar]\!]$ of formal power series in $\hbar$ with coefficients in $\mathcal{A}(E)$, endowed with the product topology; clearly $\mathcal{A}(E)\subseteq \mathcal{A}(E)[\![\hbar]\!]$. One is then able to introduce a $\star$-product on $\mathcal{A}(E)[\![\hbar]\!]$, that is, a product such that given $F\in \mathcal{A}^p(E)$, $G\in\mathcal{A}^q(E)$
\[
F \star G = F \wedge G + o(\hbar) \qquad F \star G - (-1)^{pq}G \star F = i\hbar \left\{F, G\right\}_S+o(\hbar^2)
\]
Given a Hadamard bidistribution $\omega\in\pazocal{D}^\prime(M^2, E^{\boxtimes 2})$ (see Definition \ref{Def:3.3}) and two functionals $F, G\in\mathcal{A}(E)$, we consider the fermionic functional
\[
\begin{split}
\Gamma_\omega^n(G,F)& \deq \left(\frac{i}{2}\right)^n G^{(n)} \wedge \left(\omega^{\boxtimes n} \ast F^{(n)}\right) \ \forall \ n\in\mathbb{N}\ .
\end{split}
\]
where the quantity on the right is defined component-wise as in (\ref{Eq:4.5}) and $\omega^{\boxtimes n} \ast F^{(n)}$ is defined as in (\ref{Eq:4.4}). Notice that due to the wavefront set properties of both $F^{(n)}$ and $\omega^{\boxtimes n}$ the object above is well-defined. Then we define
\be
\label{Eq:4.6}
G \star F \deq \sum_{n\in\mathbb{N}} \hbar^n \Gamma_\omega^n (G,F)\in\mathcal{A}(E)[\![\hbar]\!]\ .
\ee
This product can then be easily extended to the whole topological $\ast$-algebra of formal power series. Now, let $\mathcal{A}(S)[\![\hbar]\!]$ and $\mathcal{A}(S_G)[\![\hbar]\!]$ be the topological $\ast$-algebras of microcausal fermionic functionals associated to the free and charged Dirac fields. Suppose that the $\star_G$-product on $\mathcal{A}(S_G)[\![\hbar]\!]$ is constructed using the Hadamard bidistribution induced by the one used to define $\star$-product on $\mathcal{A}(S)[\![\hbar]\!]$, as illustrated in Proposition \ref{Prop:3.4}. We  define the \emph{classical M{\o}ller map} $\mathscr{R}_A \colon \mathcal{A}(S_G)[\![\hbar]\!]\to \mathcal{A}(S)[\![\hbar]\!]$ by considering the pullbacks induced by $\mathcal{R}_A^{\wedge p}\colon \pazocal{E}^a(M^p, S^\oplus(M)^{\boxtimes p})\to \pazocal{E}^a(M^p, S_G^\oplus(M)^{\boxtimes p})$. Now,
\begin{theorem}
\label{Th:3.5}
 $\mathscr{R}_A$ is a well-defined $\ast$-isomorphism, algebraically and topologically.
\end{theorem}
\begin{proof}
First of all, we need to show that given a functional $F\in \mathcal{A}(S_G)\subseteq\mathcal{A}(S_G) [\![\hbar]\!]$, $\mathscr{R}_A(F)$ is a well-defined functional in $\mathcal{A}(S)[\![\hbar]\!]$; that is, we need to show that 
\be
\label{Eq:3.4}
\text{WF}\left((\mathscr{R}_A(F))_u^{(n)}\right)\subseteq \Xi_n \ \mathrm{for \ all}\ u\in\mathcal{C}(S^\oplus(M))\ \mathrm{and} \ n\in\mathbb{N}.
\ee
To do so, we need to compute the wavefront set of the classical M{\o}ller map on field configurations $\mathcal{R}_A$, restricted to a continuous map $\mathcal{R}_A \colon \pazocal{D}(S^\oplus(M))\to \pazocal{D}'({S_G^\oplus(M)})$. This is defined as the wavefront set of the distribution $r_A$, where $r_A\in\pazocal{D}'(S^\oplus(M) \boxtimes {S_G^\oplus(M)})$ is obtained thanks to Schwartz's kernel theorem and satisfies
\[
\int_M \left(\mathcal{R}_A u, v\right)_G\, d\mu_g = r_A(u,v) \ \text{for every} \ u\in\pazocal{D}(S^\oplus (M)), v\in \pazocal{D}({S_G^\oplus(M)})\ .
\]
To compute its wavefront set, we proceed locally, as presented in \cite{Hormander1998}. Namely, let $\{e_i\}_{1\le i \le \text{rank}(S_G(M))}$ be a local\footnote{Notice that as we are supposing that the principal bundle $P_G$ is trivial, there exist \emph{global} frames for both $S(M)$ and $S_G(M)$.} frame for $S_G(M)$ and $\{f_j\}_{1\le j \le \text{rank}(S(M))}$ a local frame for $S(M)$; a local frame for $S^\oplus(M)$ and $S_G^\oplus(M)$ is then given by
\be
\label{Eq:4.7}
\{f_1, \dots, f_{\mathrm{rank}(S(M))}, \overline{f}_1, \dots, \overline{f}_{\mathrm{rank}(S(M))}\} \qquad \{e_1, \dots, e_{\mathrm{rank}(S_G(M))}, \overline{e}_1, \dots, \overline{e}_{\mathrm{rank}(S_G(M))}\}
\ee
respectively. Using these, we can write locally
\[
r_A = {r_A}_{ji} {f}^j \boxtimes e^i
\]
with ${r_A}_{ij}\in \pazocal{D}'(M^2)$ and where $\{e^i\}_{i}$ and $\{f^j\}_{j}$ denote the dual frames of $S_G^{\oplus}(M)^\ast \simeq \overline{S_G(M)}\oplus S_G(M)$ and $S^\oplus(M)^\ast \simeq \overline{S(M)}\oplus S(M)$ with respect to those in (\ref{Eq:4.7}). Then
\[
\text{WF}(\mathcal{R}_A) \doteq \text{WF}(r_A)=\bigcup_{\substack{1\le i \le 2\text{rank}(S_G(M)) \\ 1\le j \le 2\text{rank}(S(M))}} \text{WF}\left({r_A}_{ji}\right)
\]
where $\text{WF}({r_A}_{ji})$ is given locally by
\[
{\widehat{\phi}_\alpha}^{-1}(\text{WF}(\phi_\alpha^\ast {r_A}_{ji}))
\]
with $\widehat{\phi}_\alpha \colon T^\ast M_{U_\alpha}^2\to U_\alpha^2 \times \mathbb{R}^{\text{dim}(M^2)}$. Moreover, due to the properties of the wavefront set, we can study $\text{WF}(\mathcal{R}_A)$ by studying separately the two terms
\[
\mathcal{i}^\oplus \ \mathrm{and} \ S_-^\oplus\circ A^\oplus \circ \mathcal{i}^\oplus\ .
\]
Let us thus first consider the map $\mathcal{i}^\oplus\colon \pazocal{D}(S^\oplus(M))\to \pazocal{D}'({S_G^\oplus(M)})$; then we have that
\[
\begin{split}
\int_M \left(\mathcal{i}^\oplus u, v\right)_G\, d\mu_g & = \int_M \left(iu_1, v_2\right)_G + \left(v_1, \overline{i}u_2\right)_G\, d\mu_g \\ & = \int_M \left(u_1^\sharp(\sigma(p)), v_2^{\sharp_G}(\tilde{\sigma}(p))\right)_V\, d\mu_g
 +\int_M \left(v_1^{\sharp_G}(\tilde{\sigma}(p)), u_2^{\sharp}(\sigma(p))\right)_V\, d\mu_g \\ &= \int_M {u_1^\sharp}^j(\sigma(p))h_{ji}{v_2^{\sharp_G}}^i(\tilde{\sigma}(p))\, d\mu_g 
 + \int_M {v_1^{\sharp_G}}^i(\tilde{\sigma}(p))h_{ij} {u_2^\sharp}^j(\sigma(p))\, d\mu_g
\end{split}
\]
that is,
\[
i_{ji}=\begin{dcases}
\delta_\Delta h_{ji}\, & j \le \mathrm{rank}(S(M)), i \le \mathrm{rank}(S_G(M)) \\
\delta_\Delta \overline{h}_{(j-\mathrm{rank}(S(M)))(i-\mathrm{rank}(S_G(M))}\, & j > \mathrm{rank}(S(M)), i > \mathrm{rank}(S_G(M))
\end{dcases}\ .
\]
Therefore,
\[
\text{WF}(\mathcal{i}^\oplus) = \left\{(x,x, k, -k) \in T^\ast M^2 \setminus z(M^2)\right\}\ .
\]
As far as the map $S_-^\oplus \circ A^\oplus \circ \mathcal{i}^\oplus$ is concerned, we just need to compute the wavefront set of $A^\oplus$. In particular, proceeding as before and using our preferred gauge $\tilde{\sigma}$ we have
\[
\begin{split}
&\int_M (A u, v)_G\, d\mu_g = \int_M \left(i(g^{ij}(p)e_{j}(p)) \cdot_V \left({\rho_G}_\ast\left(\left(\sigma_G^\ast \omega^G\right)_p(e_i(p))\right)u_1^{\sharp_G}(\tilde{\sigma}(p))\right),v_2^{\sharp_G}(\tilde{\sigma}(p))\right)_V\, d\mu_g  \\
& + \int_M \left(v_1^{\sharp_G}(\tilde{\sigma}(p)),\overline{i(g^{ij}(p)e_{j}(p))} \cdot_V \left(\overline{{\rho_G}_\ast\left(\left(\sigma_G^\ast \omega^G\right)_p(e_i(p))\right)}u_2^{\sharp_G}(\tilde{\sigma}(p))\right)\right)_V\, d\mu_g \\
& = \int_M  \left(i(g^{ij}(p)e_{j}(p)) \cdot_V \left({\rho_G}_\ast\left(\left(\sigma_G^\ast \omega^G\right)_p(e_i(p))\right)\right)\right)^l_j {u_1^{\sharp}}^j(\tilde{\sigma}(p)) h_{li}{v_2^{\sharp_G}}^i(\tilde{\sigma}(p))\, d\mu_g \\
&\quad + \int_M   \left(\overline{i(g^{ij}(p)e_{j}(p)) \cdot_V \left({\rho_G}_\ast\left(\left(\sigma_G^\ast \omega^G\right)_p(e_i(p))\right)\right)}\right)^l_i {v_1^{\sharp_G}}^i(\tilde{\sigma}(p)){u_2^{\sharp_G}}^j(\tilde{\sigma}(p))h_{lj}\ .
\end{split}
\]
Thus
\[
A_{ji} = \begin{dcases}
\delta_\Delta \left(i(g^{ij}(p)e_{j}(p)) \cdot_V \left({\rho_G}_\ast\left(\left(\sigma_G^\ast \omega^G\right)_p(e_i(p))\right)\right)\right)^l_j h_{li}\, & j, i \le \mathrm{rank}(S_G(M)) \\
\delta_\Delta \left(\overline{i(g^{ij}(p)e_{j}(p)) \cdot_V \left({\rho_G}_\ast\left(\left(\sigma_G^\ast \omega^G\right)_p(e_i(p))\right)\right)}\right)^l_{i-\mathrm{rank}(S_G(M))} \overline{h}_{(j-\mathrm{rank}(S(M))l}\, & j, i > \mathrm{rank}(S_G(M))\ .
\end{dcases}
\]
Using the fact that $\sigma_G^\ast \omega^G$ is assumed to be compactly supported in $\mathrm{supp}(\mathscr{A})$, we then have that
\[
\begin{split}
\text{WF}(A^\oplus)& \subseteq \left(\pi^{-1}_{T^\ast M^2}(\mathrm{supp}(\mathscr{A}) \times M)\setminus z(M^2)\right) \cap \left\{(x,x, k, -k) \in T^\ast M^2 \setminus z(M^2)\right\} \\
& = \left\{(x, x, k, -k)\in T^\ast M^2 \setminus z(M^2) \ | \ x\in \mathrm{supp}(\mathscr{A})\right\}\ .
\end{split}
\]
We are now in a position to compute $\text{WF}(\mathcal{R}_A)$. First of all, notice that the composition $A^\oplus \circ \mathcal{i^\oplus}$ is, as we expected, well-defined: indeed, as $\supp(\mathcal{i}^\oplus) \subseteq \Delta$, we have that $\supp(\mathcal{i}^\oplus) \ni (x,y)\mapsto y$ is proper, and moreover
\[
\begin{split}
   \text{WF}'(\mathcal{i}^\oplus)_M & = \left\{(y, \xi_y) \ | \  (x,y,0, -\xi_y)\in \text{WF}(\mathcal{i}^\oplus)\right\} = \varnothing \\
   \text{WF}(A^\oplus)_M & = \left\{(x, \xi_x) \ | \  (x,y,\xi_x, 0) \in \text{WF}(A^\oplus)\right\} = \varnothing\ . 
\end{split}
\]
Thus $\text{WF}'(\mathcal{i}^\oplus)_M \cap \text{WF}(A^\oplus)_M = \varnothing$, and \cite[Theorem 8.2.14]{Hormander1998} gives us the well-posedness of $A^\oplus \circ \mathcal{i}^\oplus$ as well as
\[
\text{WF}(A^\oplus \circ \mathcal{i}^\oplus) \subseteq \left\{(x, x, k, -k) \in T^\ast M^2 \setminus z(M^2)\ | \ x \in \mathrm{supp}(\mathscr{A})\right\}\ .
\]
Let us now consider ${S_-^G}^\oplus \circ A^\oplus \circ \mathcal{i}^\oplus$; thanks to \cite{Radzikowski1996} we know that 
\[
\text{WF}({S_-^G}^\oplus) = \left\{(x,y,\xi_x, -\xi_y) \in T^\ast M^2 \setminus z(M^2)\  |\   x\in J^+(y), (x,\xi_x) \sim (y, \xi_y) \ \mathrm{or} \ x=y, \xi_x = \xi_y\right\}\ .
\]
As evidently $\text{WF}'(A^\oplus \circ \mathcal{i}^\oplus)_M = \varnothing$ and $\supp(A^\oplus \circ \mathcal{i}^\oplus)\ni (x,y) \mapsto y$ is proper, we can then apply again \cite[Theorem 8.2.14]{Hormander1998} and we have that $\text{WF}'({S_-^G}^\oplus \circ A^\oplus \circ \mathcal{i}^\oplus)$ is contained in
\[
\text{WF}'({S_-^G}^\oplus) \circ \text{WF}'(A^\oplus \circ \mathcal{i}^\oplus) \cup \left(\text{WF}({S_-^G}^\oplus)_M \times M \times \{0\}\right) \cup \left(M \times \{0\} \times \text{WF}'(A^\oplus \circ i^\oplus)_M\right)\ .
\]
The first set is given by
\[
\left\{(x, z, \xi_x, \xi_z) \ | \ \exists(y, \xi_y) \ \text{s.t.} \ (x,y, \xi_x, -\xi_y)\in \text{WF}({S_-^G}^\oplus), \ (y,z, \xi_y, -\xi_z)\in \text{WF}(A^\oplus \circ \mathcal{i}^\oplus) \right\}\ .
\]
Using the properties of $\text{WF}(A^\oplus \circ \mathcal{i}^\oplus)$ we conclude that the previous set is given by
\[
\left\{ (x,z,\xi_x, \xi_z) \ | \ (x, z, \xi_x, -\xi_z) \in \text{WF}(S_-^G), \ z\in \mathrm{supp}(\mathscr{A})\right\}\ .
\]
The second one and the third one are easily seen to be empty; therefore
\[
\text{WF}({S_-^G}^\oplus \circ A^\oplus \circ \mathcal{i}^\oplus) \subseteq \left\{(x,z,\xi_x, -\xi_z) \ | \ (x, z, \xi_x, -\xi_z)\in \text{WF}({S_-^G}^\oplus), \ z\in\mathrm{supp}(\mathscr{A})\right\}\ .
\]
Therefore 
\[
\text{WF}(\mathcal{R}_A) \subseteq \left\{(x,z,\xi_x, -\xi_z) \ | \ (x, z, \xi_x, -\xi_z)\in \text{WF}({S_-^G}^\oplus), \ z\in\mathrm{supp}(\mathscr{A})\right\} \cup \left\{(x,x, k, -k) \in T^\ast M^2 \setminus z(M^2)\right\}\ .
\]
As the map $(\mathscr{R}_A(F))_u^{(n)}$ involves the composition of $F^{(n)}_{\mathcal{R}_A u}$ with the map ${\mathcal{R}_A}^{\wedge n}$, the last step consists in computing $\text{WF}(\mathcal{R}_A^{\wedge n})$. This, according to \cite[Theorem 8.2.9]{Hormander1998}, is a subset of
\[
\begin{split}
\Big\{ & (x_1, y_1, x_2, y_2, \dots, x_n, y_n, \xi_1, \eta_1, \xi_2, \eta_2, \dots, \xi_n, \eta_n) \ | \ \exists I \subseteq \{1, \dots, n\}, I \ne \varnothing  \ \text{s.t.} \\
& (x_i, y_i, \xi_i, \eta_i) \in \text{WF}(\mathcal{R}_A) \ \forall i \in I\ \text{and}\ (x_j, y_j, \xi_j, \eta_j) = (x_j, y_j, 0,0)  \\
& \text{with} \ (x_j, y_j)\in \supp(r_A) \ \forall j \in \{1, \dots, n\}\setminus I \Big\}\ .
\end{split}
\]
Finally, we are able to check whether $\mathscr{R}_A(F)$ is well-defined. First of all, notice that $(\mathscr{R}_A(F))_u^{(n)} = F_{\mathcal{R}_A u}^{(n)}\circ {\mathcal{R}_A}^{\wedge n}$ is well-defined: indeed, thanks to \cite[Theorem 8.2.13]{Hormander1998} we know that the composition is well-defined if $\text{WF}(F^{(n)}_{\mathcal{R}_A u}) \cap \text{WF}'({\mathcal{R}_A}^{\wedge n})_{M^n}= \varnothing$; but
\[
\left\{(x_1, \dots, x_n, \xi_1, \dots, \xi_n) \ | \ (x_1, y_1, \dots, x_n, y_n, -\xi_1, 0, \dots, -\xi_n, 0) \in \text{WF}({\mathcal{R}_A}^{\wedge n})\right\} = \varnothing\ .
\]
We also infer that
\[
\text{WF}\left((\mathscr{R}_A(F))_u^{(n)}\right) \subseteq \text{WF}({\mathcal{R}_A}^{\wedge n})_{M^n} \cup \text{WF}'({\mathcal{R}_A}^{\wedge n}) \circ \text{WF}(F^{(n)}_{\mathcal{R}_A u})
\]
i.e.
\[
\begin{split}
\text{WF}\left((\mathscr{R}_A(F))_u^{(n)}\right) \subseteq \Big\{& (x_1, \dots, x_n, \xi_1, \dots, \xi_n) \ | \ (x_1, y_1, \dots, x_n, y_n, \xi_1, -\eta_1, \dots, \xi_n, -\eta_n) \in \text{WF}(\mathcal{R}_A^{\wedge n})\\
& \text{for some} \ (y_1, \dots, y_n, \eta_1, \dots, \eta_n)\in\text{WF}(F_{\mathcal{R}_A u}^{(n)})\Big\}\ .
\end{split}
\]
Using this, we shall prove (\ref{Eq:3.4}) by contradiction. Assume then that $(x_1, \dots, x_n, \xi_1, \dots, \xi_n)\in \overline{V_{x_1}^+} \times \cdots \times \overline{V_{x_n}^+}$; then there exists $(y_1, \dots, y_n, \eta_1, \dots, \eta_n)\in \text{WF}(F_{\mathcal{R}_A u}^{(n)})$ such that
\[
(x_1, y_1, \dots, x_n, y_n, \xi_1, -\eta_1, \dots, \xi_n, -\eta_n)\in \text{WF}({\mathcal{R}_A}^{\wedge n})\ .
\]
If $i\in I$, then $(x_i, y_i, \xi_i, -\eta_i) \in \text{WF}({\mathcal{R}_A})$, that is, either 
\[
(x_i, y_i, \xi_i, -\eta_i)\in \text{WF}({S_-^G}^\oplus) \  \text{with} \ y_i\in\mathrm{supp}(\mathscr{A})
\]
or
\[
x_i = y_i \qquad \eta_i =\xi_i\ .
\]
In the first case, we would then need to have $\eta_i\in \overline{V^+_{y_i}}$, as $\eta_i$ is the cotangent vector to a future-directed lightlike geodesic, while in the second one the same result follows from the equality $\eta_i = \xi_i$. If $i\notin I$, then
\[
(x_i, y_i, \xi_i, \eta_i) = (x_i, y_i, 0, 0)
\]
i.e. $\eta_i = 0 \in \overline{V^+_{y_i}}$. Thus, we conclude that $(y_1, \dots, y_n, \eta_1, \dots, \eta_n)\in \overline{V^+_{y_1}} \times \cdots \overline{V^+_{y_n}}$; but we reached a contradiction, as this is not possible by the definition of $F$. The case $(x_1, \dots, x_n, \xi_1, \dots, \xi_n)\in  \overline{V_{x_1}^-} \times \cdots \times \overline{V_{x_n}^-}$ leads to a similar conclusion: indeed, if $i\in I$ then $(x_i, y_i, \xi_i, -\eta_i) \in \text{WF}(\mathcal{R}_A)$, which as before entails that either
\[
(x_i, y_i, \xi_i, -\eta_i)\in \text{WF}({S_-^G}^\oplus) \ \text{with} \ y_i \in \mathrm{supp}(\mathscr{A})
\]
or
\[
x_i = y_i \qquad \xi_i = \eta_i\ .
\]
The first case is not possible, as we require $(x_i, \xi_i) \sim (y_i, \eta_i)$ which is not possible if $\xi_i$ is past-directed; thus from the second case we infer that $(y_1, \dots, y_n, \eta_1, \dots, \eta_n)\in \overline{V^-_{y_1}} \times \cdots \overline{V^-_{y_n}}$, which is again a contradiction.

We thus have the well-posedness of the map $\mathscr{R}_A \colon \mathcal{A}(S_G)[\![\hbar]\!]\to \mathcal{A}(S)[\![\hbar]\!]$. Recall that $R_A$ admits an inverse which is explicitly given by (\ref{Eq:2.7}); it can be shown that by defining an analogous $\widehat{\mathscr{R}}_A \colon \mathcal{A}(S)[\![\hbar]\!] \to \mathcal{A}(S_G)[\![\hbar]\!]$ we reach the same conclusion, and that $\mathscr{R}_A \circ \widehat{\mathscr{R}}_A = \mathrm{id}_{\mathcal{A}(S)[\![\hbar]\!]}$ and $\widehat{\mathscr{R}}_A \circ \mathscr{R}_A = \mathrm{id}_{\mathcal{A}(S_G)[\![\hbar]\!]}$. Therefore, $\mathscr{R}_A$ is a vector space isomorphism.

The fact that $\mathscr{R}_A$ is an algebra homomorphism is due to the following fact: we know that the $\star$-product in $\mathcal{A}(S)[\![\hbar]\!]$ is given by (\ref{Eq:4.6}); given two homogeneous functionals $F\in\mathcal{A}^p(S)$ and $G\in\mathcal{A}^q(S)$, the functional $\Gamma^n_\omega(G,F)$ appearing in the sum can be formally written, on an homogeneous element $u_1\wedge \cdots \wedge u_{p+q-2n}\in \wedge^{p+q-2n} \pazocal{E}(S^\oplus(M))$ as
\[
\begin{split}
&\Gamma^n_\omega(G,F)(u_1 \wedge \cdots \wedge u_{p+q-2n}) = \left(\frac{i}{2}\right)^n\sum_{\sigma\in S_{p+q-2n}}\mathrm{sgn}(\sigma) \int_{M^{2n}}d\mu_g(x_1)d \mu_g(y_1) \cdots d\mu_g(x_n) d\mu_g(y_n) \\
& \qquad \left(F^{(n)}_{u_{\sigma(q-n+1)} \wedge \cdots \wedge u_{\sigma(p+q-2n)}}\right)_{f_1 \cdots f_n}(y_1, \dots, y_n)\left(G^{(n)}_{u_{\sigma(1)} \wedge \cdots \wedge U_{\sigma(q-n)}}\right)_{g_1 \cdots g_n}(x_1, \dots, x_n)\\
& \qquad \omega_{s_1 t_1}(x_1, y_1) \cdots \omega_{s_nt_n}(x_n, y_n)h^{f_1s_1}h^{t_1g_1}\cdots h^{f_ns_n}h^{t_ng_n}\ .
\end{split}
\]
We are interested in computing $\Gamma_\omega^n(\mathscr{R}_A(G), \mathscr{R}_A(F))$ with $F\in\mathcal{A}^p(S_G)$ and $G\in\mathcal{A}^q(S_G)$; this amounts to substituting the formal integral above with
\[
\begin{split}
& \int_{M^{2n}}d\mu_g(\xi_1)d\mu_g(\eta_1)\cdots d\mu_g(\xi_n)d\mu_g(\eta_n) d\mu_g(x_1)d \mu_g(y_1) \cdots d\mu_g(x_n) d\mu_g(y_n)\\
& \qquad \left(F^{(n)}_{{\mathcal{R}_A}^{\wedge p-n}(u_{\sigma(q-n+1)} \wedge \cdots \wedge u_{\sigma(p+q-2n)})}\right)_{l_1 \cdots l_n}(\xi_1, \dots, \xi_n)\left(G^{(n)}_{{\mathcal{R}_A}^{\wedge q-n}(u_{\sigma(1)} \wedge \cdots \wedge u_{\sigma(q-n)})}\right)_{k_1 \cdots k_n}(\eta_1, \dots, \eta_n)\\
& \qquad {r_A}_{f_1 u_1}(x_1, \xi_1) {r_A}_{g_1 v_1}(y_1, \eta_1) \cdots {r_A}_{f_n u_n}(x_n, \xi_n) {r_A}_{g_n v_n}(y_n, \eta_n) h^{l_1 u_1}h^{v_1 k_1} \cdots h^{l_n u_n}h^{v_n k_n}\\
& \qquad \omega_{s_1 t_1}(x_1, y_1) \cdots \omega_{s_nt_n}(x_n, y_n)h^{f_1s_1}h^{t_1g_1}\cdots h^{f_ns_n}h^{t_ng_n}\ .
\end{split}
\]
By performing a simple computation one can notice that 
\[
\begin{split}
 \int_{M^2}d\mu_g(x_i)d\mu_g(y_i) {{r}_A}_{f_i u_i}(x_i, \xi_i) &{{r}_A}_{g_i v_i}(y_i, \eta_i) \omega_{s_i t_i}(x_i, y_i) h^{f_i s_i} h^{t_i g_i} \\
&= \int_{M^2}d\mu_g(x_i)d\mu_g(y_i){{r}_A}_{ u_i f_i}^\ast(\xi_i,x_i) {{r}_A}_{v_i g_i}^\ast (\eta_i, y_i) \omega_{s_i t_i}(x_i, y_i) h^{f_i s_i} h^{t_i g_i}\\
 &= {\omega_G}_{u_i v_i }(\xi_i, \eta_i)
\end{split}
\]
and therefore we arrive at
\[
\begin{split}
& \int_{M^{2n}}d\mu_g(\xi_1)d\mu_g(\eta_1)\cdots d\mu_g(\xi_n)d\mu_g(\eta_n)\\
& \qquad \left(F^{(n)}_{{\mathcal{R}_A}^{\wedge p-n}(u_{\sigma(q-n+1)} \wedge \cdots \wedge u_{\sigma(p+q-2n)})}\right)_{l_1 \cdots l_n}(\xi_1, \dots, \xi_n)\left(G^{(n)}_{{\mathcal{R}_A}^{\wedge q-n}(u_{\sigma(1)} \wedge \cdots \wedge u_{\sigma(q-n)})}\right)_{k_1 \cdots k_n}(\eta_1, \dots, \eta_n)\\
& \qquad {\omega_G}_{u_1 v_1}(\xi_1, \eta_1) \cdots {\omega_G}_{u_n v_n}(\xi_n, \eta_n)h^{l_1 u_1}h^{v_1 k_1} \cdots h^{l_n u_n}h^{v_n k_n}
\end{split}
\]
which entails that
\[
\Gamma^n_\omega(\mathscr{R}_A(G), \mathscr{R}_A(F))(u_1 \wedge \cdots \wedge u_{p+q-2n}) = \mathscr{R}_A(\Gamma^n_{\omega_G}(G,F))(u_1 \wedge \cdots \wedge u_{p+q-2n})\ .
\]
The result can be extended to the whole $\pazocal{E}^a(M^{p+q-2n}, S^\oplus(M)^{\boxtimes p+q-2n})$ by continuity, as well as to the whole algebra $\mathcal{A}(S_G)$ and thus to the whole algebra of formal power series $\mathcal{A}(S_G)[\![\hbar]\!]$. Therefore,
\[
\mathscr{R}_A(F) \star \mathscr{R}_A(H) = \mathscr{R}_A(F \star_G H)
\]
and $\mathscr{R}_A$ is an algebra homomorphism as required. 

As far as the behaviour of the classical M{\o}ller map with respect to conjugation is concerned, let us recall that on $\mathcal{A}(S)[\![\hbar]\!]$ and $\mathcal{A}(S_G)[\![\hbar]\!]$ the conjugation map is given by the natural extension of (\ref{Eq:4.8}) (with the appropriate conjugation map, that is, with either $C\colon \pazocal{E}(\overline{S(M)})\to \pazocal{E}(S(M)) $ or $C_G \colon \pazocal{E}(\overline{S_G(M)})\to \pazocal{E}(S_G(M))$) to formal power series.

Given $F\in\mathcal{A}(S_G)\subseteq \mathcal{A}(S_G)[\![\hbar]\!]$ and $u\in\mathcal{C}(S^\oplus(M))$ we thus have that
\[
\begin{split}
(\mathscr{R}_A(F))^\ast (u) & = \overline{(\mathscr{R}_A(F)) (u^\ast)} = \sum_{p\in\mathbb{N}} \overline{\braket{F_p, {\mathcal{R}_A}^{\wedge p}u_p^\ast}}\ .
\end{split}
\]
Now, if $u_p = u_{i_1} \wedge \cdots \wedge u_{i_p}$, $u_{i_j}\in\pazocal{E}(S^\oplus(M))$ we have that
\[
{\mathcal{R}_A}^{\wedge p}(u_p^\ast) = \mathcal{R}_A u_{i_p}^\ast \wedge \cdots \wedge \mathcal{R}_A u_{i_1}^\ast = (\mathcal{R}_A u_{i_p})^\ast \wedge \cdots \wedge (\mathcal{R}_A u_{i_1})^\ast = \left({\mathcal{R}_A}^{\wedge p} (u_1 \wedge \cdots \wedge u_p)\right)^\ast \ .
\]
By continuity then it holds that
\[
\sum_{p\in\mathbb{N}} \overline{\braket{F_p, {\mathcal{R}_A}^{\wedge^p}u_p^\ast}} = \sum_{p\in\mathbb{N}} \overline{\braket{F_p, ({\mathcal{R}_A}^{\wedge p}u_p)^\ast }} = (\mathscr{R}_A(F^\ast))(u)\ .
\]
Therefore, $\mathscr{R}_A\colon \mathcal{A}(S_G(M))[[\hbar]]\to \mathcal{A}(S(M))[[\hbar]]$ is a well-defined algebraic $\ast$-isomorphism. For what concerns the topological part, it is simple to understand that the sequential completeness of H\"ormander topology is satisfied, since pointwise convergence holds term by term for the formal power series and that the wave front set condition is as well satisfied by the construction seen before.
\end{proof}

\section{Conclusions and outlook}
Our classic treatment of fermions in external backgrounds has proceeded with the aim at establishing the most general framework possible for the passage to the quantum case, having in mind perturbation theory hence in the language of formal power series. In doing so we have, however, taken shortcuts that is worth mentioning again. Two of them are particularly important: the first is that we have dealt with simply connected spacetimes, which rules out topological effects (\emph{i.e.}\ inequivalent spinor structures \cite{Isham78} and Aharonov-Bohm like effects, for which see, \emph{e.g.}, \cite{Dappiaggi2020,Vasselli2019}), the second is our simplifying assumption about the compactness of the support of the gauge potentials, which rules out Coulomb potentials. Both assumptions can be relaxed, but to keep this paper into a reasonable length we postpone any further discussion. However, notice that for the first problem, the passage to Fredenhagen's universal algebra \cite{Klaus} may help to solve the issue and that, as far as the second is concerned, it is exactly due to the compactness of the \emph{spatial} support of the potentials that one rules out secular effects in perturbation theory \cite{Pinamonti2023}. 

In this paper we have privileged the \emph{pointwise} treatment of the geometric structures (sections, potentials \emph{etc}.) to guarantee a detailed and unambiguous discussion of their features. In particular, our most interesting result has been to show how the \emph{classical} M\o ller maps are algebraic and topological isomorphisms of the charged and uncharged microcausal fermionic algebras, as formal power series. Here, the use of wave front sets was essential. As for a possible next step, once the appropriate interactions are introduced, one would construct the respective \emph{quantum} M\o ller maps (see, \emph{e.g.}, \cite{Drago2017}) and proceed to the explicit computation of physical effects. One of the most ambitious aims would be, for instance, to compute the Lamb Shift for hydrogenoid atoms \cite{Eides2001} from first principles, avoiding \emph{ad hoc} assumptions and rigorously controlling eventual approximations. 

\section*{Acknowledgement} It is a pleasure to thank Nicola Pinamonti for his interest in our research.

\nocite{hawking_ellis_1973}
\nocite{Hamilton2017}
\nocite{Sanders2010}
\nocite{dappiaggi09}
\nocite{Brouder2018}
\nocite{Bar2005}
\nocite{Brunetti2019}
\clearpage
\printbibliography
\end{document}